%% file: main.tex
\documentclass[twocolumn]{article}

\usepackage{algorithm}
\usepackage[noend]{algorithmic}
\usepackage{float}
\usepackage{amsmath,amssymb,amsfonts}
\usepackage{subfig}
\usepackage{balance}
\usepackage{xcolor}
\usepackage{graphicx}
\usepackage{url}
\usepackage{multirow}
\usepackage{amsthm}
\usepackage{authblk}
\usepackage{paralist}

\newtheorem{Lemma}{Lemma}
\newtheorem{theorem}{Theorem}
\newtheorem{example}{Example}
\newtheorem{problem}{Problem}
\newtheorem{definition}{Definition}

\newcommand{\eager}{\texttt{Eager}}
\newcommand{\afro}{\texttt{pBlocking}} 

\newcommand{\bloss}{\texttt{BLOSS}}

\newcommand{\myparagraph}[1]{\smallskip \noindent \textbf{#1. }}

\usepackage[normalem]{ulem} 

\newcommand{\review}[1]{{#1}}

\newcommand{\dfnew}[1]{{{#1}}}

\newcommand{\sg}[1]{#1}
\newcommand{\df}[1]{#1}
\newcommand{\bgraph}[0]{{$P$}}

\usepackage{eqparbox}

\begin{document}

\title{Efficient and Effective ER with Progressive Blocking}
\author[1]{Sainyam Galhotra} 
\author[2]{Donatella Firmani}
\author[3]{Barna Saha}
\author[4]{Divesh Srivastava}

\affil[1]{UMass Amherst, \texttt{sainyam@cs.umass.edu}}
\affil[2]{Roma Tre University, \texttt{donatella.firmani@uniroma3.it}}
\affil[3]{UC Berkeley, \texttt{barnas@berkeley.edu}}
\affil[4]{AT\&T Labs -- Research, \texttt{divesh@research.att.com}}

\date{}                     


\maketitle

\begin{abstract}
Blocking is a mechanism to improve the efficiency of Entity Resolution (ER) 
which aims to \review{quickly prune out all}  non-matching record pairs. However, depending on the distributions of entity cluster sizes, existing techniques can be either (a) too aggressive, such that they help scale but can adversely affect the ER effectiveness, or (b) too permissive, potentially harming ER efficiency. In this paper, we propose a new methodology of \emph{progressive blocking} (\afro{}) to enable both efficient and effective ER, which works seamlessly across different entity cluster size distributions. 

\afro{} is based on the insight that the effectiveness-efficiency trade-off is revealed only when the output of ER starts to be available. Hence, \afro{} leverages partial ER output in a feedback loop to refine the blocking result in a data-driven fashion. 
Specifically, we bootstrap \afro{} with traditional blocking methods and  progressively improve the building and scoring of blocks until we get the desired trade-off, leveraging a limited amount of ER results as a guidance at every round. \review{We formally prove that \afro{}  converges}  efficiently ($O(n \log^2 n)$ time complexity, where $n$ is the total number of records). Our experiments show that incorporating partial ER output in a feedback loop can improve the efficiency and effectiveness of blocking by 5x and 60\% respectively, \review{improving the overall}  F-score of the entire ER process up to 60\%. 
\end{abstract}

\section{Introduction}
\label{sec:intro}
\sloppy
 Entity Resolution (ER) is the problem of identifying which records in a data \review{set} refer to the same real-world entity~\cite{elmagarmid2007duplicate}.
ER technologies are key for solving complex tasks (e.g., building a knowledge graph) but comparing all the record pairs to decide which pairs match is often infeasible. For this reason, the first step of ER selects sub-quadratic number of record pairs to compare in the subsequent steps. To this end, a commonly used approach is \emph{blocking}~\cite{papadakis2016comparative}. Blocking groups similar records into \emph{blocks} and then selects pairs from the ``cleanest'' blocks -- i.e., those with fewer non-matching pairs -- for further comparisons.
The literature is rich with methods for building and processing blocks~\cite{papadakis2016comparative}, but depending on the data \review{set} at hand, different techniques can either leave too many matching pairs outside, leading to incomplete ER results and low effectiveness, or include too many non-matching pairs, leading to low efficiency. 
\begingroup
\setlength{\tabcolsep}{2pt} 
\begin{table}
\scriptsize
\centering
\caption{Sample records (we omit schema information) referring to 4 distinct entities. $r_i^{e}$ represents the i-th record referring to entity $e$. Records in the first \review{two} rows refer to a Chevrolet Corvette C6 ($c6$) and a Z6 ($z6$). Records in the \review{last two rows} to a Chevrolet Malibu ($ma$) and a Citr{\"o}en C6 ($ci$) (same model name as Corvette C6 but different car). 
\label{tab:example}}
\begin{tabular}{|l|l|}
\hline
$r_1^{c6}$: chevy corvette c6       & $r_2^{c6}$: chevy corvette c6 navigation \\ \hline
$r_3^{c6}$: chevrolet corvette c6 & $r_1^{z6}$: corvette z6 navigation \\ \hline
$r_1^{ma}$: chevy malibu navigation & $r_2^{ma}$: chevrolet chevy malibu       \\ \hline
$r_3^{ma}$: chevrolet malibu     & $r_1^{ci}$: citroen c6 navigation  \\ \hline
\end{tabular}
\end{table}
\endgroup

\begin{figure*}[t]
    \centering
    {\subfloat[\label{fig:mypipel}]{\includegraphics[height=2.3cm]{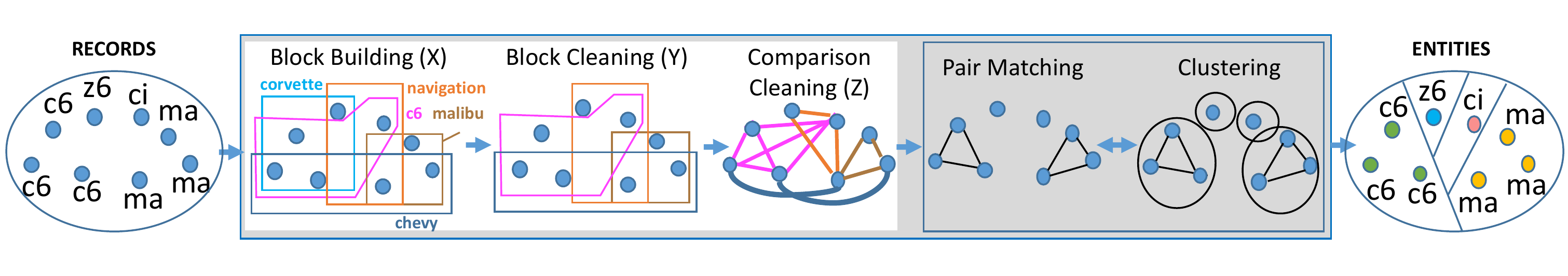}}}
    {\subfloat[\label{fig:distr}]{\includegraphics[height=2.3cm]{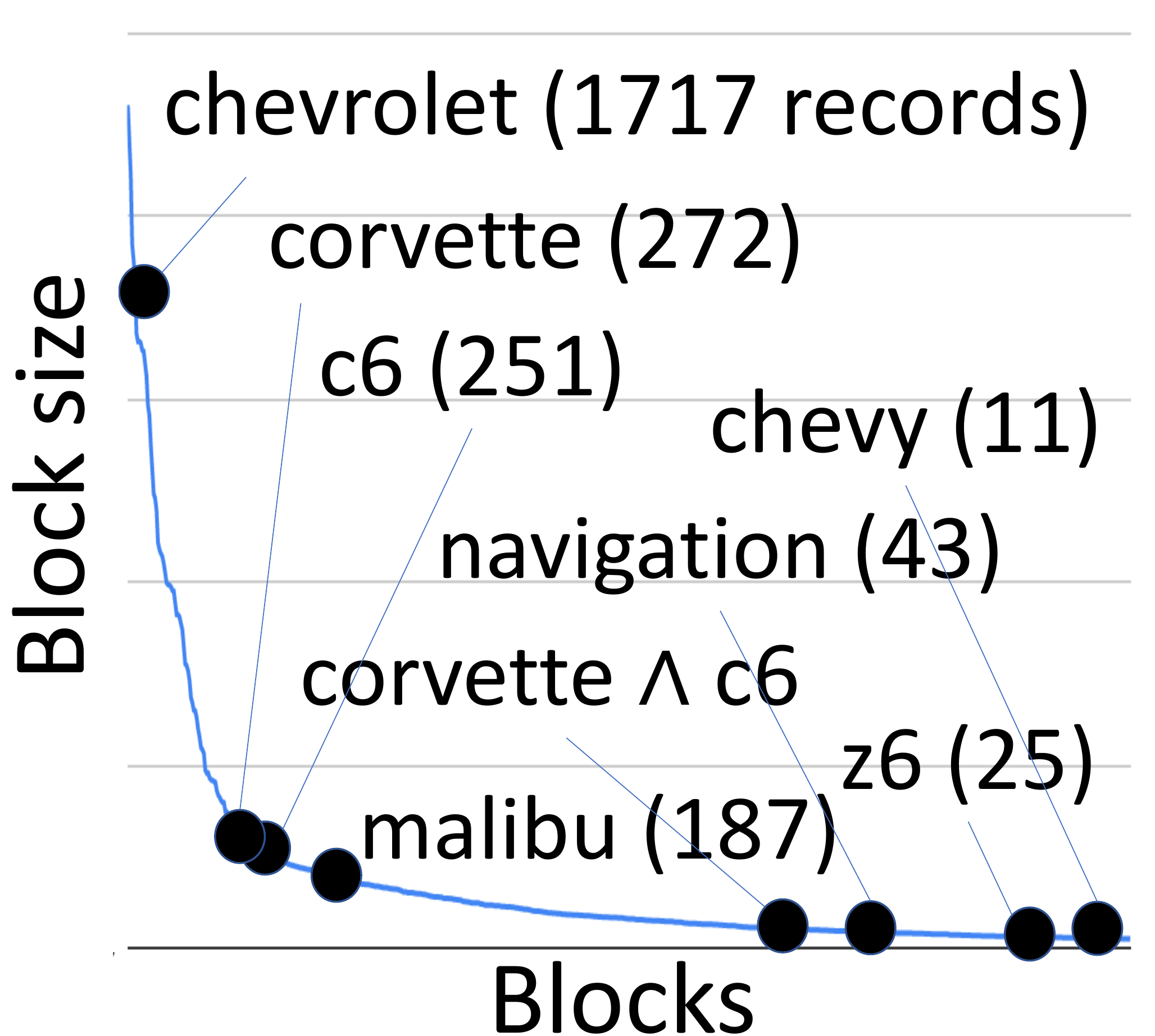}}}
    \caption{\dfnew{(a) Illustration of a standard blocking pipeline. Block building, block cleaning and comparison cleaning sub-tasks are highlighted in white. The downstream ER algorithm is shown in gray. Description of each record is reported in Table~\ref{tab:example}. 
    (b) Block size distribution (standard blocking) for the real \texttt{cars} dataset used in our experiments.}}
\end{figure*}



\myparagraph{\dfnew{\afro{}}} We propose a new \emph{progressive} blocking technique that overcomes the above limitations by short-circuiting the two operations -- blocking and pair comparisons -- that are traditionally solved sequentially. 
Our method starts with an aggressive blocking step, which is efficient but not very effective. Then, it  computes  a limited amount of ER results on \review{a subset of pairs selected by the aggressive blocking, and sends these partial (matching and non-matching) results from the ER phase back to the blocking phase, creating a ``loop'', to improve blocking effectiveness.} In this way, \dfnew{blocking can progressively self-regulate and} adapt to the properties of each dataset, with no configuration effort.  We illustrate our blocking method, that we call \afro{}, in the following example. 

\begin{example} 
Consider the records in Table~\ref{tab:example} from the \texttt{cars} dataset used in our experiments, \dfnew{and a standard schema-agnostic blocking strategy $\mathcal{S}$ such as~\cite{papadakis2015schema}. As shown in Figure~\ref{fig:mypipel}, we consider three blocking sub-tasks~\cite{papadakis2016comparative}. First, during \emph{block building}, $\mathcal{S}$ creates a separate block for each text token (we only show the blocks `corvette', `navigation', `malibu', 'c6' and `chevy').  Then, during \emph{block cleaning}, {$\mathcal{S}$  uses a threshold to prune out all the blocks of large size}. Depending on the threshold value (using the block sizes in the entire {\tt cars} dataset, shown in Figure 1b), we can have any of the following extreme behaviors. (Note that no intermediate setting of the threshold can yield a sparse set of candidates that is at the same time complete.)}
\begin{compactitem} 
    \item \emph{Aggressive} blocking: \dfnew{$\mathcal{S}$ prunes every block except the smallest one (`chevy') and returns} $(r_1^{c6},r_2^{c6})$, $(r_1^{c6},r_1^{ma})$, $(r_2^{c6},r_1^{ma})$ and $(r_1^{ma},r_2^{ma})$, missing $r_3^{c6}$ and $r_3^{ma}$.
    \item \emph{Permissive} blocking: \dfnew{$\mathcal{S}$ prunes only the largest block (`chevrolet') and returns many non-matching pairs}.
\end{compactitem}
\dfnew{Finally, during \emph{comparison cleaning}, $\mathcal{S}$ can use another threshold to further prune out pairs sharing few blocks, e.g. by using \emph{meta-blocking}~\cite{metablocking}. As in block cleaning, different threshold values can yield aggressive or permissive behaviours. Note that matching pairs such as $(r_2^{c6},r_3^{c6})$ share the same number of blocks (`corvette' and `c6') as non-matching pairs such as  $(r_2^{c6},r_1^{z6})$ (`corvette' and `navigation'). (Even worse, `c6' is larger than `navigation'.)} 
\label{ex:motivcars}
\end{example}

\afro{} can solve these problems in a few rounds: the first round does aggressive blocking, the second round does more effective blocking \dfnew{by making targeted updates} accordingly to partial ER results, and so on. \dfnew{Examples of such updates to the blocking result are discussed below.}
\begin{compactenum}
    \item \review{Creation of new blocks that help \emph{inclusion} of} $(r_1^{c6},r_3^{c6}), (r_2^{c6},r_3^{c6})$: 
    \review{\afro{}} creates a \emph{new} block `corvette $\land$ c6' \dfnew{with records present in both blocks `corvette' and `c6'.} \review{This block is much smaller than its two constituents and has only Corvette C6 cars.}
    \item \dfnew{Adaptive cleaning to help \emph{inclusion} of} $(r_1^{ma},r_3^{ma}), (r_2^{ma},r_3^{ma})$: 
    \dfnew{\afro{} can discourage pruning of block `malibu' that contains Chevrolet Malibu cars, even if it is a large block;}
    \item \dfnew{Adaptive cleaning to help \emph{exclusion} of non-matching pairs}: 
    \dfnew{\afro{} can encourage pruning of block `navigation' that contains no matching pairs, even if it is a small block.} 
\end{compactenum}
After a few rounds of updates like the above, \review{\afro{} returns all the matching pairs with very few non-matching pairs}. Note that after the last round, the ER output can be computed on the resulting pairs as in the traditional setting. Updates of type (1) are performed via a new \emph{block intersection} algorithm, while (2) and (3) are performed by a new \emph{block scoring} method. By construction, when the blocking scores converge, the entire blocking result also converges. 



\myparagraph{Our contributions} The main contribution of this paper is a new blocking methodology with both high efficiency and effectiveness in a variety of application scenarios. Since \afro{} can in principle start off using any blocking strategy, it represents not only a new approach but also a way to ``boost'' traditional ones. \afro{} works seamlessly across different entity cluster size distributions such as:
\begin{compactitem}
    \item \emph{small entity clusters}, where, using block intersection, \afro{} can recover entities such as Corvette C6 consisting of few records sharing large and dirty blocks.
    \item \emph{large entity clusters}, where, using block scoring, \afro{} can recover entities such as Chevrolet Malibu consisting of many records sharing large and clean blocks.
\end{compactitem}

\noindent We prove theoretically and show empirically that, with a few rounds and a limited amount of \review{partial ER} results, our progressive blocking method can provide a significant boost in blocking effectiveness without penalizing efficiency. Specifically, we (i) demonstrate fast convergence and low space and time complexity ($O(n \log^2 n)$, where $n$ is the number of records) of \afro{}; (ii) report experiments achieving up to 60\% increase in recall when compared to state-of-the-art blocking~\cite{dal2018bloss}, and up to 5x boost in efficiency. Finally, we observe that \afro{} can yield up to 70\% increase on the F-score of the final ER result, thus confirming the substantial benefits of our approach.

\myparagraph{Outline} The rest of this paper is organized as follows. Sections~\ref{sec:prel} and~\ref{sec:arch} provide preliminary discussions and a high-level description of the \afro{} approach. Sections~\ref{sec:refinement} and~\ref{sec:feedback} explain  our block intersection and block scoring methods, respectively. 
Section~\ref{sec:perf} provides theoretical  analysis of \afro{}'s effectiveness  and Section~\ref{sec:exp} provides extensive experimental results and key takeaways. Section~\ref{sec:related} discusses the related work and we conclude in Section~\ref{sec:concl}. 


\section{Blocking Preliminaries}
\label{sec:prel}

\begin{table}
\centering
\scriptsize
\caption{Notation Table}
\begin{tabular}{||c | c||} 
 \hline
 $V$ & Collection of records \\ \hline
 $\mathcal{C}$ & Collection of clusters \\ \hline
 $B$ & Block: A subset of records, $B \subseteq V$\\\hline
 $p_m(u,v)$ & Similarity between $u$ and $v$\\\hline
 \bgraph{}~$=(V,A')$ & Blocking graph, $A' \subset V \times V$ \\\hline
 $\phi$ & Feedback frequency\\\hline
 $p(B)$ & Probability score of a block $B$ \\\hline
 $u(B)$ & Uniformity score of block $B$\\ \hline
 $H(B)$ & Entropy  of block $B$\\ \hline
 $\mathcal{H}$ & Block Hierarchy \\ \hline
 $G_t$ & Random Geometric graph \\
 \hline
 $\gamma$ & Fraction of nodes used for scoring blocks\\ \hline
 $\mu_g$& Expected similarity of a matching edge\\\hline
 $\mu_r$& Expected similarity of a non-matching edge\\\hline
\end{tabular}
\label{table:notation}
\end{table}

\df{ 
{Table~\ref{table:notation} summarizes the main symbols used throughout this paper.}} \dfnew{Let $V$ be the input set of records, with $|V|=n$. Consider an (unknown) graph $\mathcal{C}=(V, E^{+})$, where $(u,v) \in E^{+}$ means that $u$ and $v$ represent the same entity.
$\mathcal{C}$ is transitively closed, that is, each of its connected components $C \subseteq V$ is a clique representing a distinct entity. We call each clique a \textit{cluster} of $V$, and refer to the partition induced by $\mathcal{C}$ as the ER \emph{ground truth}.}

\dfnew{\begin{definition}[Pair Recall] 
Given a set of matching record pairs $A' \subseteq V \times V$, Pair Recall is the fraction of pairs $(u,v) \in E^+$ that can be either (i) matched directly, because  $(u,v)\in A'$, or (ii) indirectly inferred from other pairs $(u,w_0), (w_0,w_1), \dots, (w_c,v) \in A'$ by connectivity.
\end{definition}
}

\dfnew{A formal definition of the blocking task follows.}

\dfnew{\begin{problem}[Blocking Task]
Given a set of records $V$, group records into possibly overlapping blocks $\mathcal{B}\equiv \{B_1, B_2, \dots\}$, $B_i\subseteq V$ and compute a graph \bgraph{}~$=(V,A')$, where $A' \subseteq A$, $A \equiv \{(u,v) : \exists B_i \in \mathcal{B}\textnormal{ s.t. }u \in B_i \land v \in B_i\}$, such that $A'$ is sparse ($|A'| << {n \choose 2}$) and $A'$ has high Pair Recall. We refer to \bgraph{} as the \emph{blocking graph}.
\label{prob:blocking}
\end{problem}} 

\dfnew{The blocking graph \bgraph{} is the final product of blocking and contains all the pairs that can be considered for pair matching. The efficiency and effectiveness of the blocking method is measured as  Pair Recall (PR) of (the set of edges in) \bgraph{} and the number of edges in it for a certain PR, respectively.}
%
Blocking methods consist of three sub-tasks as defined by~\cite{papadakis2016comparative}: block building, block cleaning and comparison cleaning. In the following, we describe each of these steps and the corresponding methods in the literature. 

\emph{Block building} ($\mathcal{BB}$) takes as input $V$ and returns a block collection $\mathcal{B}$, by assigning each record in $V$ to possibly multiple blocks. The popular \emph{standard blocking}~\cite{papadakis2015schema} strategy creates a separate block $B_t$ for each token $t$ in the records and assigns to $B_t$ all the records that contain the token $t$. In order to tolerate spelling errors, \emph{q-grams blocking}~\cite{gravano2001approximate} considers \review{character-level} q-grams instead of entire tokens. Other strategies include \emph{canopy clustering}~\cite{mccallum2000efficient} \review{and} \emph{sorted neighborhood}~\cite{hernandez1995merge}. \review{Canopy clustering iteratively selects a random seed record $r$, and creates a new block $B_r$ (or a canopy) with all the records that have a high similarity with $r$ with respect to a given similarity function (e.g., using a subset of features~\cite{mccallum2000efficient}). We can use different similarity functions to build different sets of canopies.} Sorted neighborhood sorts all the records \review{according to multiple sort orders (e.g., each according to a different attribute~\cite{hernandez1995merge}) and then it slides a window $w$ of tokens over each ordering, every time creating a new block $B_w$. Blocks have the same number of distinct tokens but the number of records in a block can vary significantly. \df{Each {of} these techniques creates $O(n)$ blocks.}}

\emph{Block cleaning} ($\mathcal{BC}$) takes as input \dfnew{the} block collection $\mathcal{B}$ and returns a subset $\mathcal{B}' \subseteq \mathcal{B}$ by \dfnew{pruning} blocks that may contain too many non-matching record pairs. Block cleaning is typically performed by assigning each block a $score: \mathcal{B} \rightarrow {\rm I\!R}$ with a block scoring procedure and then \dfnew{pruning} blocks with low score. Traditional scoring strategies include functions of block sizes such as TF-IDF~\cite{elmagarmid2007duplicate,papadakis2012blocking}.

\emph{Comparison cleaning} ($\mathcal{CC}$) takes as input the set $A$ of all the intra-block record pairs in \dfnew{the} block collection $\mathcal{B}'$ \dfnew{(which is a subset of the intra-block record pairs in $\mathcal{B}$)} and returns a graph \bgraph{}~$=(V,A')$, with $A' \subseteq A$, by \dfnew{pruning} pairs that \review{are likely to} be non-matching. Comparison cleaning is typically performed by assigning each pair a $weight: A \rightarrow {\rm I\!R}$ and then pruning pairs with low weight. Weighting strategies include \emph{meta-blocking}~\cite{metablocking} possibly with active learning~\cite{simonini2016blast,dal2018bloss}. 
In classic meta-blocking, $weight(u,v)$ corresponds to the number of blocks in which $u$ and $v$ co-occur, \review{based on the assumption that that more blocks a record pair shares, the more likely it is to be matching.\footnote{\review{This assumption holds for block building methods such as standard blocking, q-grams blocking and sorted neighborhood with multiple orderings~\cite{hernandez1995merge}, and extends naturally to canopy clustering by using multiple similarity functions.}}} The recent \bloss{} strategy~\cite{dal2018bloss} employs active learning on top of the pairs generated by meta-blocking, and learns a classifier using features extracted from the blocking graph for further pruning. 

\df{We denote with $\mathcal{B}$($X,Y,Z$) a blocking strategy that uses the methods $X$, $Y$, and $Z$, respectively for block building, block cleaning and comparison cleaning. The strategy used in our $\texttt{cars}$ example (Example~\ref{ex:motivcars}) can be thus denoted as $\mathcal{B}$(\emph{standard blocking, TF-IDF, meta-blocking}).}

\myparagraph{After blocking} \dfnew{Typical ER algorithms include \emph{pair matching} and entity~\emph{clustering} operations. Such operations label as ``matching'' the pairs referring to the same entity and ``non-matching'' otherwise, and typically require the use of a classifier~\cite{mudgal2018deep} or a crowd~\cite{wang2013leveraging}. Clustering consists of building a possibly noisy clustering  $\mathcal{C}'$ according to labels, and can be done with a variety of techniques, including robust variants of connected components~\cite{VerroiosErrors} and random graphs~\cite{galhotraSigmod}. This noisy clustering is the final product of ER. 
}

\section{Overview of \MakeLowercase{p}Blocking}
\label{sec:arch}


Analogous to traditional blocking methods, \afro{} takes as input a collection $V$ of records and returns a blocking graph \bgraph{}. 
%
A high-level view of the methods introduced in \afro{}, for each of the main blocking sub-tasks of Section~\ref{sec:prel}, is provided below. 
Such methods, unlike previous ones, can leverage a feedback of partial ER results. 

\emph{Block building} \review{in \afro{}} constructs new blocks \sg{arranged in the form of} a \emph{hierarchy}. First level blocks are initialized with blocks generated by a traditional method (e.g., standard blocking, sorted neighborhood, canopy clustering or q-gram blocking). Subsequent levels contain intersections of the blocks in the previous levels. 
\afro{} can use feedback from the \review{partial} ER output to build intersections such as `corvette $\land$ c6' that can lead to new, cleaner blocks, and avoid bad intersections such as `corvette $\land$ chevrolet' that would not improve the fraction of matching pairs in \bgraph{} (Chevrolet Corvette C6 and Z6 are different entities). We discuss block intersection in Section~\ref{sec:refinement}.

\emph{Block cleaning} \review{in \afro{}} prunes dirty blocks based on feedback-based scores. First round scores \review{are} initialized with a traditional method \review{(e.g. TF-IDF)}. Then, scores are \review{refined} based on feedback by combining two quantities: the fraction $p(B)$ of matching pairs in a block $B$, and the block uniformity $u(B)$, which captures the distribution of entities within the block ($u(B)$ is the inverse of \emph{perplexity}~\cite{manning1999foundations}). \sg{Since the goal of blocking phase is to identify blocks \review{that} have a higher fraction of matching pairs and fewer entity clusters, we combine the above values as $score(B) = p(B) \cdot u(B)$.}  
\afro{} can use feedback from the \review{partial} ER output to estimate 
$p(B)$ and $u(B)$, yielding high scores for clean blocks such as `malibu' (high $p(B)$ and high $u(B)$) and low scores for dirtier blocks such as `navigation' (low $p(B)$ and low $u(B)$), and `c6' (low $u(B)$). We discuss block scoring in Section~\ref{sec:feedback}. 

Finally, \emph{comparison cleaning} \review{in \afro{}} is implemented with a traditional method such as meta-blocking.


\begin{algorithm}[t]
\caption{Our blocking method \afro{}}
\label{algo:afro}
\begin{algorithmic}[1]
{\scriptsize
\REQUIRE Records $V$, methods $X$, $Y$, $Z$ for each blocking step. Default: \emph{X=standard blocking, Y= TF-IDF} and \emph{Z=meta-blocking}.
\ENSURE Blocking graph \bgraph{}
\STATE $\mathcal{C}' \leftarrow \emptyset$ 
\STATE $B\leftarrow$ build the first level of block hierarchy with method $X$
\STATE $scores \leftarrow$ initialize block scores using method $Y$
\STATE \bgraph{} $ \leftarrow$ block cleaning and comparison cleaning with method $Z$
\STATE \bgraph{}$_{new} \leftarrow \emptyset$
\FOR{round=2; round $\leq 1 / \phi$ $\land$ \bgraph{} $\neq$ \bgraph{}$_{new}$; round++}
    \WHILE{ER progress is less than $\phi$} 
    \STATE $\mathcal{C}' \leftarrow$ \df{Execute an incremental step of method $W$ for pair matching and clustering on \bgraph{}}
    \ENDWHILE

    \STATE $score \leftarrow$ update the block scores according to $\mathcal{C}'$ \texttt{  //Feedback}
    \STATE $B\leftarrow$ update the block hierarchy based on $score$
    \STATE \bgraph{} $ \leftarrow $ \bgraph{}$_{new}$
    \STATE \bgraph{}$_{new} \leftarrow$ block cleaning and comparison cleaning with $Z$
\ENDFOR
\RETURN H
}
\end{algorithmic}
\end{algorithm}

\myparagraph{Workflow} \sg{Algorithm~\ref{algo:afro} describes the \afro{} workflow and how the introduced blocking methods can be used. We denote with \afro{}($X,Y,Z$) \dfnew{a progressive blocking strategy} that uses the methods $X$, $Y$ and $Z$, respectively for building the first level of the block hierarchy, initializing the block scores, and performing comparison cleaning \dfnew{as described in Algorithm~\ref{algo:afro}}. In our \texttt{cars} examples, we have \afro{}(\emph{standard blocking, TF-IDF, meta-blocking}). 

We first initialize the set of clusters $\mathcal{C}'$, the block hierarchy and the block scores (\textbf{lines 1--3}). The next step (\textbf{line 4}) consists of computing the first version of the blocking graph \bgraph{} according to the selected method for comparison cleaning (e.g., meta-blocking). The graph \bgraph{} is then progressively updated, round after round (\textbf{lines 6--12}). In order to activate the feedback mechanism, \afro{} needs to interact with an ER algorithm $W$ for pair matching and clustering operations (\textbf{line 7--8}).} \df{ Algorithm \review{$W$} is executed over \bgraph{} until it makes a \emph{progress} of $\phi$ with $\phi \in [0,1]$, that is, until {$\phi \cdot n \log^2{n}$} record pairs have been processed since the previous round.\footnote{For algorithms such as~\cite{vesdapunt2014crowdsourcing}, progress can be defined as a fraction $\phi \cdot n$ of processed \emph{records} since the previous round.} At that point, the algorithm $W$ is interrupted, $\mathcal{C}'$ is updated (\textbf{line 8}) and sent as feedback to all of \afro{}'s components.} 
Based on such feedback, we update the function $score(B)=p(B) \cdot u(B)$ (\textbf{line 9}) and construct new blocks in the form of a hierarchy (\textbf{line 10}). Higher score blocks are used to enumerate the most promising record pairs and generate the updated blocking graph \bgraph{}$_{new}$ (\textbf{lines 11-12}). When either the maximum number of rounds $\frac{1}{\phi}$ has been reached (setting $\phi=1$ is the same as switching off the feedback) or the blocking result converges (\bgraph{} $=$ \bgraph{}$_{new}$), \afro{} terminates by returning \bgraph{}. 


\review{We present a formal analysis of the effectiveness of \afro{} in Section~\ref{sec:perf}.} We refer to Section~\ref{sec:exp} for experiments. Due to its robustness to different choices of the pair matching algorithm $W$, we do not include $W$ in \afro{}'s parameters (differently from $X$, $Y$, $Z$). Natural choices for $W$ include \emph{progressive} ER strategies that can process \bgraph{} in an online fashion and compute $\mathcal{C}'$ incrementally~\cite{verroioswaldo,vesdapunt2014crowdsourcing,mudgal2018deep}. However, traditional algorithms, such as~\cite{elmagarmid2007duplicate} can be used as well by adding \emph{incremental ER} techniques~\cite{gruenheid2014incremental,incrementwang} on top.

\subsection{Computational complexity}
\label{sec:compl}

For efficiency, it is crucial to ensure that the total time and space taken to compute \bgraph{} is close to linear in $n$. Since every round of \afro{} comes with its own time and space overhead, we first describe how to bound the complexity of every round and then discuss how to set the parameter $\phi$ in Algorithm~\ref{algo:afro} (and thus the maximum number of rounds) so as to bound the complexity of the entire workflow.  

\myparagraph{Round Complexity} \afro{} implements the following strategies to decrease overhead of each round. 

\emph{Efficient block cleaning.} We compute the block scores by sampling $\Theta(\log n)$ \sg{records from each of the top $O(n)$} high-score blocks computed in the previous round. 

\emph{Efficient comparison cleaning.} For simplicity, we build \bgraph{} by enumerating at most $\Theta(n\log^2 n)$ intra-block pairs by processing blocks in non-increasing block score. 

Based on the above discussion, we have Lemma~\ref{lem:compl}.

\begin{Lemma}
A single round of \afro{}($X,Y,Z$), such as \afro{}(\emph{standard blocking, TF-IDF, meta-blocking}) has $O(n\log^2 n)$ \sg{space and time complexity}.
\label{lem:compl}
\end{Lemma}
\begin{proof}
We first show that the total feedback is limited to $O(n\log^2 n)$ \sg{space complexity}, even though it considers all transitively inferred matching and non-matching edges, which can be $\Omega(n\log^2 n)$. 
For the matching pairs, we store all the records  with an entity id such that any pair of records that have been resolved share the same id. This requires $O(n)$ space in the worst case and captures all the matching edges that have been identified in the ER output. For the non-matching pairs, we store a non-matching edge between their entity ids. Since the maximum number of pairs returned by \afro{} is limited to $O(n\log^2 n)$, the total number of pairs compared in each round and thus the number of non-matching edges stored is also $O(n\log^2 n)$.  
Then, we analyze the complexity of using feedback for the $\mathcal{BB}$ and $\mathcal{BC}$ tasks. Since the maximum number of blocks considered in any round for the scoring component is $O(n)$ and the scoring mechanism samples $O(\log^2 n)$ pairs from each block, the total number of edges enumerated for block scoring and building is $O(n\log^2 n)$. Since the maximum number of pairs for inclusion in the graph $H$ is also $O(n\log^2 n)$, a single round of \afro{} outputs $H$ in $O(n\log^2 n)$ total work.
\end{proof}
\myparagraph{Workflow Complexity} \review{As discussed in Section~\ref{sec:perf}, $\phi$ can be set to a small constant fraction.}
Thus, along with Lemma~\ref{lem:compl}, this guarantees an $O(n\log^2 n)$ complexity for the entire workflow. Experimentally a smaller $\phi$ value yields higher final recall,  thus as a default we set $\phi=0.01$, yielding a maximum of $100$ rounds. Although such a $\phi$ value gets the best trade-off between effectiveness and efficiency in our experiments, we also observe that slight variations of its setting do not affect the performance much (Section~\ref{sec:exp}), demonstrating the robustness of \afro{}.

\section{Block Building}
\label{sec:refinement}
\input{refinement.tex}

\section{Block Cleaning}
\label{sec:feedback}

Let $A' \subset V \times V$ be the pairs selected by blocking phase at a given point (we recall that $A'$ is the edge set of the blocking graph \bgraph{}~$=(V,A')$) and each considered pair $(u,v)\in A'$ has a similarity value denoted by $p_m(u,v)$. A block $B\subseteq V$ refers to a subset of records. Using this notation, we discuss the different methods for scoring blocks and how the scores converge with feedback for effective ER performance.


\input{scoring.tex}
Finally, for the presented analysis, we assumed that oracle answers are correct. Nonetheless, (i) for small amount of oracle errors ($\sim 5\%$), we can leverage methods such as~\cite{galhotraSigmod,VerroiosErrors} to correct them, and (ii) in more challenging applications with up to $20\%$ erroneous answers, we show experimentally (see Section~\ref{sec:exp}) that \afro{} keeps converging, only at a slightly slower rate and demonstrates robustness. 

%

\section{Experiments}
\label{sec:exp}
\input{experiments.tex}



\section{Related work}
\label{sec:related}

Blocking has been used to scale Entity Resolution (ER) for a very long time. However, all the techniques in the literature have considered blocking as a preprocessing step and suffered from the trade-off between effectiveness and efficiency/scalability. 
We divide the related work into two parts: advanced blocking methods which we improve upon, and progressive ER methods which can be used to generate a limited amount of matching/non-matching pairs to send as a feedback to our blocking computation. 


\myparagraph{Advanced blocking methods} There are many blocking methods in the literature with different internal functionalities and solving different blocking sub-tasks. In this paper, we considered four representative block building strategies, namely standard blocking~\protect\cite{papadakis2015schema}, canopy clustering~\protect\cite{mccallum2000efficient}, sorted neighborhood~\protect\cite{hernandez1995merge} and q-grams blocking~\protect\cite{gravano2001approximate}. It is well-known that such techniques can yield a fairly dense blocking graph when used alone. We refer the reader to~\cite{papadakis2016comparative} for an extensive survey of various blocking techniques and their shortcomings. \dfnew{Such block building strategies can be used as the  method $X$ in our Algorithm~\ref{algo:afro}.}

Recent works have proposed advanced methods that can be used in combination with the mentioned block building techniques by focusing on the comparison cleaning sub-task (thus improving on  efficiency). 
The first technique in this space is  \emph{meta-blocking}~\cite{metablocking}.  Meta-blocking \dfnew{aims to extract the most similar pairs of records by leveraging block-to-record relationships and} can be very efficient in reducing the number of unnecessary pairs produced by traditional blocking techniques, but \review{it is} not always easy to configure. To this end, follow-up works \dfnew{such Blast~\cite{simonini2016blast} use ``loose'' schema information to distinguish promising pairs, while~\cite{bilenko2006adaptive} and SNB~\cite{papadakis2014supervised} rely on a sample of labeled pairs for learning accurate blocking functions and classification models respectively. Finally,} the most recent strategy \bloss{}~\cite{dal2018bloss} uses \emph{active learning} to select such a sample \dfnew{and configure the meta-blocking}. 
The goal of traditional  meta-blocking~\cite{metablocking} and its follow-up techniques like \bloss{}~\cite{dal2018bloss} prune out low similarity candidates from the blocking graph generated using various block building strategies discussed above. Their performance is highly dependent on the effectiveness of block building techniques and the quality of blocking graph. On the other hand, \afro{} constructs meaningful blocks that effectively capture majority of the matching pairs and scores each block based on their quality to generate   fewer non-matching pairs in the blocking graph. Meta-blocking techniques compute the blocking graph statically, prior to ER, and thus can be used as the $Z$ method in our Algorithm~\ref{algo:afro}.  In Figure~\ref{fig:bbafro} we compare with classic meta-blocking and \bloss{}, as the latter shows its superiority over Blast and SNB.

\myparagraph{Progressive ER} Many applications need to resolve data sets efficiently but do not require the ER result to be complete.  Recent literature described methods to compute the best possible partial solution. Such techniques include \emph{pay-as-you-go} ER~\cite{whang2013pay} that use ``hints'' on records that are likely to refer to the same entity and more generally \emph{progressive} ER such as the schema-agnostic method in~\cite{simonini2018schema} and the strategies in~\cite{altowim2014progressive}\cite{papenbrock2015progressive} that consider a limit on the execution time. In our discussion, we considered oracle-based techniques, namely \texttt{Node}~\cite{vesdapunt2014crowdsourcing}, \texttt{Edge}~\cite{wang2013leveraging}, and \texttt{Eager}~\cite{galhotraSigmod}. Differently from other progressive techniques, oracle-based methods consider a limit on the number of pairs that are examined by the oracle for matching/non-matching response. Such techniques were originally designed for dealing with the crowd but they can also be used with a variety of classifiers due to their flexibility.  All these techniques naturally work in combination with \afro{} by sending as feedback their partial results. 

\myparagraph{Other ER methods} In addition to the above methods, we mention works on ER architectures that can help users to debug and tune parameters for the different components of ER~\cite{gokhale2014corleone,das2017falcon,konda2016magellan,papadakis2018return}. Specifically, the approaches in~\cite{gokhale2014corleone,das2017falcon} show how to leverage the crowd in this setting. All of these techniques are orthogonal to the scope of our work and we do not consider them in our analysis. 
\df{The previous work in~\cite{whang2009entity} proposes to greedily merge records as they are matched by ER, while processing the blocks one at a time. Each merged record (containing tokens from the component records) is added to the unprocessed blocks, permitting its participation in the subsequent matching and merging by their iterative algorithm. 
Limitations of processing blocks one at a time has been shown in more recent blocking works~\cite{metablocking}. 
} 
\section{Conclusions}
\label{sec:concl}
We have proposed a new blocking algorithm, \afro{} that progressively updates the relative scores of blocks and constructs new blocks by leveraging a novel feedback mechanism from partial ER results. Most of the techniques in the literature perform blocking as a preprocessing step to prune out redundant non-matching record pairs. However, these techniques are sensitive to the distribution of cluster sizes and the amount of noise in the data set and thus are either highly efficient with poor recall or have high recall with poor efficiency. \afro{} can boost the effectiveness and efficiency of blocking across all data sets by jump-starting blocking with any of the standard  techniques and then using new robust feedback-based methods for solving blocking sub-tasks in a data-driven way. To the best of our knowledge, \afro{} is the first framework where blocking and pair matching components of ER can help each other and produce high quality results in synergy.


\bibliographystyle{abbrv}
\bibliography{main}


\appendix
\input{appendix}

\end{document}

%% file: refinement.tex
One of the major challenges of \df{block building ($\mathcal{BB}$)} is that when generating candidate pairs that capture matches it can also generate a number of non-matching pairs. This phenomenon is highly prevalent in datasets with very few matching pairs. \df{To overcome this challenge, our \emph{block building by intersection} algorithm takes a collection of \review{blocks $B_1, \ldots, B_m$} built by a traditional method for $\mathcal{BB}$ and creates \review{smaller clean blocks out of large dirty} ones, thus contributing to the recall of the blocking graph without adding extra non-matching pairs.} 
An \emph{intersection block hierarchy} $\mathcal{H}$ is constructed as follows. Let the first layer be $B_1, \dots, B_m$. Then blocks in layer $L$ consist of the intersection of $L$ distinct blocks in the first layer.


\begin{example}
\df{Consider our \texttt{cars} example in Section~\ref{sec:intro}, and the blocks corresponding to tokens  `corvette' and  `c6', namely $B_\texttt{corvette}$, and $B_\texttt{c6}$. A sample block in the second level of $\mathcal{H}$ is $B_{corvette,c6} = B_{corvette} \cap B_{c6}$. When we build the new block, we only include records containing the two tokens `corvette' and  `c6' (possibly non consecutively), thus obtaining a cleaner block than the original ones.}
\end{example}

\myparagraph{\dfnew{Refined blocks}} We refer to the newly created block as a \emph{refined} block, and to the intersecting blocks as \emph{parent} blocks. Not all the refined blocks are useful. We need one of the following correlation based conditions to hold to decide if a refined block $B_{i,j}$ must be kept in $\mathcal{H}$. 
\begin{compactitem}
\item  $score(B_{i,j})> score(B_i) \cdot score(B_j)$, that is the score of the refined block is higher than the combined score of the parent blocks.
\item \review{The existence of a randomly chosen record $r$ in blocks $B_i$ and $B_j$ is positively correlated, i.e.  $Pr[r\in B_{i,j}] = |B_{i,j}|/n > Pr(r\in B_i) \cdot Pr(r\in B_j) $, which simplifies to} $|B_{i,j}| > \frac{|B_i||B_j|}{n}$. 
\review{For example, the number of common records in blocks corresponding to tokens `c6' and `corvette'  is much higher than the common records in blocks corresponding to `navigation' and `c6'.}
\end{compactitem} 

\begin{algorithm}[t]
\caption{Block Layers Creation}
\label{algo:creation}
\begin{algorithmic}[1]
{\scriptsize
\REQUIRE Set of records $V$, depth $d$
\ENSURE 	Layer set $\{L_1,\ldots, L_d\}$
 \FOR{$i=1;i\leq d;i++$}
    \STATE $L_i\leftarrow \phi$
 \ENDFOR
 \STATE processed $\leftarrow \phi$
\FOR{$v\in V$}
\STATE blockLst$\leftarrow$ getBlocks(v)
    \FOR{$i=2;i<$d$ ;i++$}
        \FOR{\review{$\mathcal{B}= \{B_j: B_j\in $ {blockLst}\}, $|\mathcal{B}|=i$}}
        \STATE \review{$B'=\cap_{B_j\in \mathcal{B}} B_j$}
        \IF{$B'\notin $ processed}
            \STATE $L_{i}$.append($B'$)
            \STATE processed.append($B'$)
        \ENDIF
        \ENDFOR
        \STATE blockLst$\leftarrow L_i$
    \ENDFOR
\ENDFOR
}
\end{algorithmic}
\end{algorithm}

\begin{algorithm}[t]
\caption{ Layer Cleaning}
\label{algo:cleaning}
\begin{algorithmic}[1]
{\scriptsize
\REQUIRE Layer set $\{L_1,\ldots, L_d\}$
\ENSURE Cleaned Layer set $\{L_1,\ldots, L_d\}$
\FOR{$i=2;i<d;i++$}
    \FOR{block $\in L_i$}
        \STATE parentLst $\leftarrow$ getParents(block)
        \IF{$\prod_{p\in parentLst} score(p) < score(block)$ \OR $\prod_{p\in parentLst} \frac{|L_{i-1}[p]|}{n}<\frac{|L_i[block]|}{n}$}
        \STATE continue
        \ELSE
        \STATE $L_i$.remove(block)
        \ENDIF
    \ENDFOR
\ENDFOR
}
\end{algorithmic}
\end{algorithm}

\noindent Suppose the maximum depth of the hierarchy is $d$ which is a constant. The construction of \dfnew{refined} blocks can take $O(n^d)$ time if the number of blocks considered in the first layer is $O(n)$.  For efficiency, we iterate over the records (linear scan) and for each record $r$, we consider all pairs of blocks that contain $r$ as candidates to generate blocks in the different levels of the hierarchy. 
The following lemma bounds the total number of refined blocks across the hierarchy. 
\begin{Lemma}
The number of blocks present in $\mathcal{H}$ is $O(n)$ if each record $r$ is present in a constant number of blocks. 
\end{Lemma}
\begin{proof}
Our algorithm considers each record $u\in V$ and generates intersection blocks by performing conjunction of blocks that contain the record $u$. Suppose the record $u$ is present in $\gamma_u$ blocks in the first layer. Then the maximum number of blocks present in $\mathcal{H}$ that contain $u$ is $\sum_{i=1}^{d} {\gamma_u \choose i} $. Assuming $\gamma_u$ is a constant, the maximum number of blocks in the hierarchy is $n\sum_{i=1}^{d} {\gamma_u \choose i}=O(n)$.
\end{proof}

\myparagraph{Refinement algorithm} We are now ready to describe \afro{}'s intersection method for building the block hierarchy. Our method has two steps:
\begin{compactitem}
\item (Alg. \ref{algo:creation}) The first step creates all possible blocks considering the intersection search space.
\item (Alg. \ref{algo:cleaning}) The cleaning phase removes the blocks that do not satisfy the correlation criterion described above. 
\end{compactitem}
Algorithm~\ref{algo:creation} describes the creation step, which iterates over all the records in the corpus and creates all possible blocks per record. \review{The list of all blocks to which a record belongs is constructed (denoted by blockLst) and the new blocks are added in different layers. The layer of the new block depends on the number of intersecting blocks that constitute the new block}.  Then, the cleaning step in Algorithm~\ref{algo:cleaning} iterates over the different layers and keeps only the blocks that satisfy the score or size requirements. For a block in layer $q$, \texttt{getParents}() identifies the two blocks which are in layer $(q-1)$ whose conjunction generates the block being considered.  If these parents have been removed during the cleaning phase, then their parents are considered and the process is continued recursively until we end up at the ancestors present in the list of blocks. 

\review{
Block Layers Creation (Alg.~\ref{algo:creation}) constructs all the blocks in the form of a hierarchy and Layer Cleaning (Alg.~\ref{algo:cleaning}) deactivates the blocks that do not satisfy the correlation requirements. \dfnew{Since the result of block layers creation does not change in different \afro{} iterations, decoupling the creation component from the cleaning component (which changes dynamically) allows for more efficient computation.}} 

\myparagraph{Time complexity} Assuming the depth of the hierarchy is a constant, Algorithms~\ref{algo:creation} and~\ref{algo:cleaning} operate in time linear in the number of records $n$. 
Block refinement takes $3$ minutes for a data set with $1M$ records in our experiments.

%% file: scoring.tex
\myparagraph{\dfnew{Block scoring}} Block scoring helps to distinguish informative blocks based on their ability to capture records from a single cluster. By selecting pairs within informative blocks, down-stream ER operations can focus on records pairs that have high probability of being a match. 
The most common mechanism used in the literature is TF-IDF and it assigns block scores inversely proportional to the block size prioritizing smaller blocks over larger ones. If the data set has small clusters, such a simple method can work well. \review{However}, if the data set has a skewed cluster size distribution, some large blocks are just uninformative (and are rightfully less preferred by TF-IDF), but others can represent a large cluster and thus should stand out in the scoring. Distinguishing these blocks before pair matching can be difficult, but \afro{} provides a way to leverage the feedback. 



Specifically, the scoring algorithm of \afro{} prioritizes blocks having (a) high fraction of matching pairs \sg{measured as  matching probability within a block} and (b) fewer number of clusters (especially larger clusters)  measured as uniformity (a function of entropy of the cluster distribution within a given block $B$). Lower entropy and hence lower diversity values indicate the representativeness of $B$ towards a particular cluster as opposed to higher entropy values which refer to the presence of many fragmented clusters. 

\sg{More formally, the matching probability score identifies the probability that a randomly chosen pair $(u,v) \mid u,v\in B$ refers to the same entity and is defined as follows. }

\begin{definition}[Matching Probability score $p(B)$]
\sg{The value $p(B)$ is defined as the fraction of matching pairs within a block $B$.}
\end{definition}

The block uniformity, $u(B)$ captures perplexity of cluster distribution within $B$ measured in terms of its entropy.
\begin{definition}[Cluster Entropy $H(B)$]
\sg{The cluster entropy of a block, $H(B)$ refers to the entropy of the cluster distribution when restricted to the records present in  block $B$. Mathematically, $H(B)= -\sum_{C\in \mathcal{C}} p_C\log p_C $,
where $p_C=|C\cap B|/|B|$ refers to the probability that a randomly chosen node from $B$ belongs to cluster $C$.}
\end{definition}
Using $H(B)$, block uniformity score is defined as follows.
\begin{definition}[Block Uniformity $u(B)$]
\sg{The block uniformity $u(B)=e^{-H(B)}$ is the inverse of perplexity \cite{manning1999foundations} of the cluster distribution within the block where perplexity refers to the exponential of cluster distribution entropy. }
\end{definition}

\dfnew{
\begin{example}
Suppose that we know that a block $B$ contains records of two clusters $C_1$ and $C_2$ and thus we can compute the uniformity of $B$ exactly. If the two clusters are perfectly balanced in $B$, i.e., $|C_1 \cap B|= 0.5 \cdot |B|$ and $|C_2 \cap B|= 0.5 \cdot |B|$, the entropy is $H(B)=-0.5\log 0.5 - 0.5 \log 0.5 \approx 0.69$ and thus $u(B) = e^{-H(B)} = 0.5$. If there is some skew, e.g. $|C_1 \cap B|= 0.7 \cdot |B|$ and $|C_2 \cap B|= 0.3 \cdot |B|$, then the entropy is lower $H(B)=-0.7\log 0.7 - 0.3 \log 0.3 \approx 0.61$ and the uniformity is higher $u(B) \approx 0.54$. In the extreme case where $C_1 \cap B = B$ and $C_2 \cap B = \emptyset$, $H(B)=0$ and $u(B)=1$. 
\end{example}}

\noindent \dfnew{Note that when resolving two duplicate-free datasets where all clusters are of size 2 (also known as Record Linkage) the entropy increases with block size, thus block uniformity yields comparable results to traditional TF-IDF.}

{

Since the goal of block scoring is to identify blocks that have high matching probability and high uniformity, we multiply the two values to get a final estimate of the block score. 
\begin{definition}[Block Score, $score(B)$]
The score of a block $B$, $score(B)$, is defined as the product of matching probability score and uniformity score of $B$. That is, $score(B)=p(B) u(B)$.
\end{definition}
Next, we describe the algorithm to estimate these components of block score. The exact value of matching probability and block uniformity requires complete ER results. However, \afro{}  estimates these scores initially with the similarity estimates of every pair of records and refines these scores with  additional feedback from \review{partial ER results}. 

\myparagraph{Matching probability score}
The matching probability score is estimated as the average matching similarity of pairs of records within the block, i.e.: 
$$p(B)=\frac{\sum_{u,v\in B} p_m(u,v)}{{|B| \choose 2}}$$ where $p_m(u,v)$ is estimated as follows:}
\begin{compactitem}
    \item for pairs declared as matches, we set $p_m(u,v)=1$;
    \item for pairs declared as non-matches, we set $p_m(u,v)=0$;
    \item for unlabelled pairs, we use the $p_m$ values computed by common similarity metrics (e.g. via jaccard similarity or the similarity-to-probability mapping as in~\cite{papenbrock2015progressive}).
\end{compactitem}

\myparagraph{Block uniformity estimation}
Estimating uniformity score \review{requires the cluster size distribution in $B$, which is harder to infer from the prior similarity values}. We next describe a mechanism to estimate entropy $H(B)$ needed to compute the uniformity score.  
We consider each record $u \in B$, and consider the cluster $C_u$ that contains $u$. We are interested in computing $\frac{|C_u \cap B|}{|B|}$ in order to compute entropy $H(B)$. Instead, we compute the expected size of $|C_u\cap B|$ as $E_u=E[|C_u \cap B|]= \sum_{v\in B} p_m(u,v)$
based on $p_m$ values \review{of edges incident on $u$}. We compute the expected cluster size for every record $u \in B$ and sort them in non-increasing order. Let $L$ be the sorted list. Let the first record in the sorted list $L$, that is, the node with highest expected cluster size in $B$ be $u$. On expectation $u$ has $E_u$ records in $B$ that belong to $C_u$. All these records must have similar expected cluster sizes as well. 
We put $u$ and the next $\lfloor E_u \rfloor$ records from $L$ to a set $S_U$, assuming that they belong to the same cluster $C_u$. We recurse on $L \setminus S_U$ until a partition $\{S_U, S_V, \dots\}$ of the block is generated. The size of each part\review{ition} can be thought of as a rough estimate of the true cluster distribution in $B$ and is used to calculate the entropy. 

\dfnew{
\begin{example}
Consider a block $B$, with $|B|=10$. Let $[u_1, u_2 \dots u_{10}]$ be the corresponding list $L$ of records sorted in non-increasing $E_{u_i}$ values.
If $E_{u_1}=\sum_{i \in 2 \dots 10} p_m(u_1,u_i) = 6.6$ we set $S_{U1}=\{u_1 \dots u_{1+\lfloor E_{u_1} \rfloor}\}=\{u_1 \dots u_{7}\}$ and then consider the next node in $L$ which is $u_8$. If $E_{u_8}=\sum_{i \in 9, 10} p_m(u_8,u_i) = 2$ we set $S_{U8}=\{u_8 \dots u_{8+\lfloor E_{u_8} \rfloor}\}=\{u_8 \dots u_{10}\}$ and then finish. As $|S_{U1}|=0.7 \cdot |B|$ and $|S_{U8}|=0.3 \cdot |B|$ we estimate $u(B)=e^{-0.7\log 0.7 - 0.3 \log 0.3} \approx 0.54$.
\end{example}
}

The value returned by this mechanism is generally an \review{under-estimate} of the true entropy $H(B)$ but in practice it can approach $H(B)$ quickly with increasing feedback data and turns out to be very efficient. Section~\ref{sec:theory} discusses this convergence rate in different application scenarios.



\myparagraph{Efficient block cleaning} \dfnew{Traditional scoring strategies such as TF-IDF are based on block size computation and thus operate in linear time. Computing our $score(B)$ values requires instead to process intra-block pairs and thus yields potentially quadratic computation. Hence, we sample $\Theta( \log n)$ records from each block for its score computation. This strategy operates in $\Theta( \log^2 n)$ time and takes less than $1$ minute for a data set with $1M$ records in our experiments. Our sampling strategy gives  an approximation within a factor of $(1+\epsilon)$ of the matching probability scores estimated using all the records within each block (Lemma~\ref{lem:sampling}).}

\section{ Analysis of \MakeLowercase{p}Blocking}
\label{sec:perf}
In this section we present a theoretical analysis of the effectiveness of \afro{}. We first analyze the pair recall of blocking in the absence of feedback by considering a natural generative model for block creation. Next we analyze the effect of feedback on block scoring and the final recall. 
\subsection{Pair Recall without Feedback}

\dfnew{We start by giving the following basic lemma below.
\begin{Lemma}
The blocking graph \bgraph{}~$=(V,A')$ contains a spanning tree for each clique $C$ of $\mathcal{C} = (V, E^+)$  iff the Pair Recall is 1.\label{lem:rec}
\end{Lemma}
\begin{proof}
If $A'$ contains a spanning tree for each clique $C$, then any pair $(u,v)\in A'\cap E^+$ contributes directly to the recall.
\sg{All pairs of records $(u,v)$ that refer to the same entity, $(u,v) \in E^+$ and are not present in $A'$, $(u,v)\notin A'$ can be inferred from the edges in the spanning tree using transitivity, ensuring Pair Recall = 1. For the converse, let us assume that $\exists~C\in \mathcal{C}$ such that $A'$ does not contain any spanning tree over the matching edges. This implies that $C$ is split into multiple components (say $C_1$, $C_2$) when restricted to $ A'\cap E^+$ edges. In this case, the collection of matching edges joining these components, $\{(x,y), \forall x\in C_1, y\in C_2\}$ cannot be inferred as none of these edges are processed by the mentioned ER operations, yielding pair recall of \bgraph{} less than 1.} 
\end{proof}}


Our probabilistic model for block creation is motivated by the standard blocking~\cite{papadakis2015schema}, sorted neighborhood~\cite{hernandez1995merge} and canopy clustering~\cite{mccallum2000efficient}  algorithms which aim to generate blocks that capture high similarity candidate pairs. 
This model of block generation is closely related to  Random Geometric Graphs \cite{penrose2003random} which were proposed by Gilbert in 1961 and have been used widely to analyze spatial graphs.
\begin{definition}[Random Geometric Graphs]
Let $S^t$ refer to the surface of a t-dimensional unit sphere, $S^t\equiv \{x\in \mathbb{R}^{t+1}\mid ||x||_2=1 \}$.
A random geometric graph $G_t(V,E)$ of $n$ vertices $V$, has parameters $t\in \mathbb{Z}^+ $ and a real number $r\in [0,2]$. It assigns each vertex $i\in V$ to a point chosen independently and uniformly at random  within $S^t$ and any pair of vertices $i,j\in V$ are connected if the distance between their respective points is less than $r$.
\end{definition} 
Now, we define the probabilistic block generation model.
\begin{definition}[Probabilistic Block Generation] The block generation model places the records $u\in V$ independently and uniformly at random within $S^t$. Every record $u$ constructs a ball of volume $(\alpha\log n/n)$ with $u$ as the center, where $\alpha$ is a given parameter and all points within the ball are referred to as block $B_u$.
\end{definition}

The set of points present within a ball $B_u$ can be seen as high similarity points that would have been chosen as blocking candidates  in the absence of feedback. Our probabilistic block generation model constructs $n$ blocks, one for each node and every pair of records that co-occur in a block $B_u, u\in V$, has an edge in the blocking graph \bgraph{}$_g(V,E)$ (subscript $g$ to emphasize generative model). Next we analyze  pair recall of \bgraph{}$_g(V,E)$. 

\myparagraph{Notation} Let $d(u,v)$ refer to the distance between records $u$ and $v$ and $r_\epsilon$ refer to the radius of an $\epsilon$-volume ball\footnote{ $\epsilon = O(r_\epsilon^t)$.} in $t$ dimensions.
Under these assumptions we first show that the expected number of edges in the blocking graph \bgraph{}$_g$ is at least $\frac{\alpha (n-1)\log n}{2}$ and then that \bgraph{}$_g(V,E)$ has recall $< <1$.

\begin{Lemma}
The blocking graph \bgraph{}$_g(V,E)$ contains at least $\alpha\frac{(n-1)\log n}{2}$ candidate pairs on expectation.
\end{Lemma}
\begin{proof}
Each record $u\in V$, constructs a spherical ball of volume $\alpha \log n/n$, with $u$ as the center and all points within the ball are added as neighbors of $u$ in the blocking graph. Hence, the number of expected neighbors of $u$ within the ball is $\alpha (n-1)\log n/n $. 
There are a total of $n$ such blocks (one ball per record) and  each of the candidate pairs $(u,v)$ is counted twice (once for the block  $B_u$ and once for the block  $B_v$). Hence there are a total of $\frac{\alpha(n-1)\log n}{2}$ such candidate pairs.  Notice that this analysis ignores the candidate pairs $(u,v)$ which are more than $r_{\alpha\log n/n}$ from each other but are connected in the blocking graph. This would happen if they are present together in another block centered at $w\in V\setminus\{u,v\}$, that is $\exists w \mid d(u,w)\leq r_{\alpha\log n/n}$ and $d(v,w)\leq r_{\alpha\log n/n}$. This shows that the total number of candidate pairs in the blocking graph is atleast $\frac{\alpha(n-1)\log n}{2}$.
\end{proof}

 Additionally, \bgraph{}$_g(V,E)$ has the following property:
 \begin{Lemma}
 A blocking graph \bgraph{}$_g$ is a subgraph of a random geometric graph $G_t$ with $r=2r_{\alpha\log n/n}$
 \end{Lemma}
 \begin{proof}
 Following the construction of blocking graph, if the distance between any pair of vertices $u,v\in V$ is less than or equal to $r_{c\log n/n}$, then $(u,v)\in E$. Similarly, any pair of nodes $u,v\in V$ such that $d(u,v)> 2r_{c\log n/n}$, then $(u,v)\notin E$. However, if $r_{c\log n/n}<d(u,v)\leq 2r_{c\log n/n}$, the pair $(u,v)\in H_g$ only if $\exists w\in V$ such that $d(u,w)\leq r_{c\log n/n}$  and $d(v,w)\leq r_{c\log n/n}$. This shows that the blocking graph $H_g$ is a subgraph of a random geometric graph where a pair of vertices (u,v) is connected only if the distance  $d(u,v)\leq 2r_{c\log n/n}$ is connected.
 \end{proof}
 
  This means that if $G_t$ has suboptimal recall then  \bgraph{}$_g$ also has poor recall and hence, we analyze the recall of $G_t$ with $r=2r_{\alpha\log n/n}$. Lemma \ref{lem:rec} shows that the blocking graph will achieve  recall $=1$ only if it contains a spanning tree of each cluster.  Hence, we analyze the formation of spanning trees in $G_t'=G_t(V,E\cap E^+)$ that refers to $G_t$ restricted to matching edges. We show the following result,
\sg{
\begin{Lemma}\label{lem:disconn}
The graph $G_t$ restricted to matching edges  in the ground truth, $E^+$ splits a cluster $C$, where $|C|=o(n/\alpha)$ into multiple components.
\end{Lemma}
\begin{proof}
Using the connectivity result from \cite{penrose2003random}, a random geometric graph $G_t$ of $n$ nodes is disconnected if the expected degree of the nodes is $< \log n$. Additionally, it splits the graph $G_t$ into many smaller clusters. Therefore, a cluster $C\in V$ is disconnected in $G_t'=G_t(V,E\cap E^+)$ if the degree of each vertex is $< \log |C|$.

The expected degree of a record $u\in C$, restricted to $G_t'$ is $O(|C|(\frac{\alpha \log n}{n} ))  =o (\log n) $ if $|C| =o(n/\alpha)$. Hence, the expected degree of each node within a cluster $C$ is $o(\log |C|)$, leading to formation of disconnected components within $C$.
\end{proof}

\begin{theorem}
A blocking graph \bgraph{}$_g(V,E)$, generated according to the probabilistic block model has recall $<1$ unless all clusters have size $\Theta(n)$ assuming $\alpha$ is a constant.
\end{theorem}
\begin{proof}
Lemma~\ref{lem:disconn} shows that the cluster $C$ of size $< n/\alpha$ is split into various disconnected components when restricted to matching edges. Hence, the blocking graph \bgraph{}$_g$ does not form a spanning tree of $C$ and will have  recall less than $1$ \review{(Lemma~\ref{lem:rec})}. Since the cluster $C$ is broken into many small clusters, the drop in recall is also significant. 
\end{proof}

\myparagraph{Remark} The analysis extends when considering less noisy data such as when only a constant fraction of records are placed randomly on the unit sphere, and the remaining records are grouped together according to the cluster identity they belong to. 
Our analysis  exposes the lack of robustness of performing blocking without feedback.

}

\subsection{Pair Recall with Feedback} 
\label{sec:theory}
\sloppy
In this section we analyze the pair recall of blocking when employed with \afro{}. For this analysis we consider the noisy edge similarity model $p_m(u,v)$ that builds on the edge noise model studied in prior work on ER \cite{firmani2016online}. 
\begin{definition}[Noisy edge model] Noisy edge model defines the similarity of a pair of records with parameters $\theta\in (0,1)$, $\beta = \Theta(\log n)$ and $\beta' = \Theta(\log n)$. A matching edge $(u,v)\in E^+$ has a similarity distributed uniformly at random within $[\theta,1]$  with probability $1-\frac{\beta}{n}$ and remaining edges are distributed uniformly within $[0,\theta)$. 
A non-matching edge has similar distribution on similarity values with $\beta'$ instead of $\beta$.
\end{definition}

When $\beta << \beta'$, the matching probability score of a block with higher fraction of matching edges is much higher than the one with fewer matching edges and  \afro{}  algorithm will consider blocks in the correct ordering even in the absence of feedback. However, it is most challenging when non-matching edges are generated with a distribution similar to matching edges, that is $\beta$ and $\beta'$ are close. We define a random variable $X(u,v)$ to refer to the edge similarity distributed according to the noisy edge model. Following this notion, let $\mu_g$ and $\mu_r$ denote the expected similarity  of a matching and non-matching edge respectively. 
$$\mu_g = \left(1-\beta/n\right)\frac{1+\theta}{2} + \frac{\beta}{n}\frac{\theta}{2}$$ and $\mu_r$ has the same value with $\beta'$ instead of $\beta$.

We show that the feedback based block score initialized with TF-IDF weights is able to achieve perfect recall with  a feedback of $\Theta(n\log^2 n)$  pairs assuming that the ER phase { makes no mistakes on the pairs that it processes, helping to ensure  the correctness of partially inferred entities. } Additionally, 
the feedback from the ER phase is distributed randomly across edges within a block. 
We also discuss the extension  when feedback is  biased towards pairs from large entity clusters and high similarity pairs. In those scenarios, \afro{}'s scoring mechanism converges quicker leveraging the larger feedback due to transitivity.

\myparagraph{Effect of Sampling}
First, we show that sampling $\Theta(\log n)$ records from a block gives  approximation within a factor of $(1+\epsilon)$ of the matching probability score computed using all the records.

\review{\begin{Lemma}
For a block $B$ with $|B| > c\log n$, the matching probability score of $B$ estimated by sampling $\Theta(\log n/\epsilon^2) $ records randomly is within $[(1-\epsilon),(1+\epsilon)]$ factor of $p(B)$ with a probability of $1-o(1)$, where $p(B)$ is the score using all $|B|$ records.
\label{lem:sampling}
\end{Lemma}}
\begin{proof}
\review{Consider a block $B$ with more than $c\log n$ records. Let $X(u,v)$ denote the edge similarity of a pair $(u,v)$ according to the noisy edge model. The matching probability score of $B$ on considering the complete block is $\frac{1}{{|B|\choose 2}}\sum_{u,v\in B} X(u,v)$. The expected score of the block ($\mu_B$) is}

\begin{eqnarray*}
\frac{1}{{|B|\choose 2}}E\left[\sum_{u,v\in B} X(u,v)\right] &&=\frac{1}{{|B|\choose 2}} \sum_{\substack {u,v\in B,\\(u,v)\in E^+}} E[X(u,v)]\\
&& + \frac{1}{{|B|\choose 2}} \sum_{\substack{ u,v\in B,\\(u,v)\in E^-}} E[X(u,v)]\\
&=&  (1-\alpha)\mu_g + \alpha \mu_r
\end{eqnarray*}
where $\alpha$ is the fraction of non-matching pairs in the block $B$.

\review{For a sample of $S=c\log n/\epsilon'^2$ records, the expected probability score ($\mu_S$) is $ (1-\alpha)\mu_g + \alpha \mu_r$, where $\epsilon'=\epsilon/(2+\epsilon)$}
\begin{eqnarray*}
\frac{1}{{c\log n \choose 2}}E[\sum_{u,v\in S} X(u,v)] &=& \frac{1}{{c\log n\choose 2}}\sum_{\substack{u,v\in S, \\(u,v)\in E^+}} E[X(u,v)]\\
&& + \frac{1}{{c\log n\choose 2}}\sum_{\substack{u,v\in S,\\ (u,v)\in E^-}} E[X(u,v)]\\
&=&  (1-\alpha)\mu_g + \alpha \mu_r
\end{eqnarray*}

\review{Using Hoeffding's inequality~\cite{hoeffding},

\begin{eqnarray*}
Pr\left[\frac{1}{{c\log n \choose 2}}\sum_{u,v\in S} X(u,v) \leq  (1-\epsilon')\mu_S \right]\\
\leq e^{-2\epsilon'^2\mu_S^2{c\log n \choose 2}}\\
\leq e^{-2\log n}=\frac{1}{n^2}\\
\end{eqnarray*}
Using the same argument, 
we can show that $Pr\left[{(1-\epsilon')}\mu_S  \leq \frac{1}{{c\log n \choose 2}}\sum\limits_{u,v\in S} X(u,v) \leq  (1+\epsilon')\mu_S \right] \geq 1- \frac{2}{n^2}$
This shows that the calculated probability score on the samples $S$ is within a factor of $(1-\epsilon')$  and $(1+\epsilon')$ of the expected score with a probability of $1-o(1)$. The probability score of $B$ on considering all records, is also within a factor of $(1-\epsilon')$  and $(1+\epsilon')$ of the expected value $\mu_S$. Therefore, the estimated score on sampling guarantees approximation within a factor of $\frac{1+\epsilon'}{1-\epsilon'} = 1+2\epsilon'/(1-\epsilon') = 1+\epsilon$ with a high probability. }
\end{proof}

\noindent \dfnew{The above lemma can extend to block uniformity because $p_m$ values are used analogously for expected cluster sizes. In Lemma~\ref{lem:atleast} we show how to set the constant within the $\Theta$ notation based on level of noise in the $p_m$ values.}

To prove the convergence of \afro{}, we first estimate the lower and upper bound of  matching probability scores of a block $B$ in the presence of feedback   and show that a feedback of $\Theta(\log^2 n)$ is enough to rank blocks with larger fraction of matching pairs higher than the blocks with fewer matching pairs. Our analysis first considers the blocks containing more than $\gamma \log n$  records  (where $\gamma$ is a large constant say $12$) and we analyze the smaller blocks separately.

\myparagraph{Convergence for large blocks}
First, we evaluate the converged block scores with a feedback $F$ and evaluate the condition that the block scores are in the correct order.
For this analysis, we consider the fraction of matching edges for block score computation but similar lemmas extend for the uniformity score calculation.

\begin{Lemma}
For all blocks $B$,  with more than $\gamma \log n$ records, the matching probability score of $B$, $p(B)$ after a feedback of $F=O(\log^2 n)$ randomly chosen pairs is at most $(1-\alpha){|F|}/{{\gamma \log n \choose 2 }} + 1.5p'(1-|F|/{\gamma \log n \choose 2})$  with a probability of $1-1/n^3$, where $\alpha$ is the fraction of non-matching pairs in $B$, $\gamma$ is a constant and $p'=\mu_g(1-\alpha) + \mu_r\alpha$.\label{lem:atleast}
\end{Lemma}
\begin{proof}
For block scoring, \afro{} considers a sample of $S= \gamma \log n$ records (where $\gamma$ is a large constant) and considers the sample ensuring that feedback $F\subseteq S\times S$  belongs to this sample. The total number of matching edges which have been identified with feedback over randomly chosen pairs is $(1-\alpha)|F|$. Let $X(u,v)$ be a random variable that refers to the similarity of the pair $(u,v)$ and $\mu(u,v)$ to its expected value.
For $S=\gamma\log n$,  the expected similarity of non-feedback edges within $C$ is 
\begin{eqnarray*}
\sum_{\substack{u,v\in S,\\ (u,v)\notin F}}\mu(u,v) &=& \sum_{(u,v)\in E^+} E[X(u,v)] + \sum_{(u,v)\notin E^+}E[X(u,v)]\\
 &=& \sum_{(u,v)\in E^+} \mu_g + \sum_{(u,v)\notin E^+}\mu_r\\
&=& \left({\gamma \log n \choose 2 }-|F|\right) (\mu_g(1-\alpha)+\mu_r\alpha)
\end{eqnarray*}

We use the Hoeffding inequality to bound the total similarity, $\sum X(u,v)$ of $T=\left({\gamma \log n \choose 2 }-|F|\right) = \gamma' {\log n \choose 2 }$, for some constant $\gamma'$, edges which do not have feedback.
\begin{eqnarray*}
\sum_{u,v\in B_c, (u,v)\notin F} X(u,v) \le  (1+\delta)\sum_{u,v\in B_c, (u,v)\notin F}\mu(u,v) 
\end{eqnarray*} with a probability of $1-e^{-2\delta^2\mu_T^2/ |T|}$ which can be simplified as $1-e^{-\delta^2\mu_T}$, since $\mu_r,\mu_g>1/2$
Hence, the probability of success simplifies to $> 1- 1/n^3$ after substituting $\delta = 0.5$.
Hence, the similarity score of the block $B$ is atmost $\left(\frac{|F|}{{\gamma\log n\choose 2}}(1-\alpha) + 1.5p'(1-|F|/{\gamma \log n \choose 2})\right)$ with a high probability.
\end{proof}

Similarly, we prove a lower bound on block score.
\begin{Lemma}
For all blocks $B$ with  $|B|\ge \gamma \log n$, the matching probability score after a feedback $F=O(\log^2 n)$ record pairs in $B$ is at least $(1-\alpha){|F|}/{{\gamma \log n \choose 2}} + 0.5 p'(1-|F|/{\gamma \log n \choose 2})$ with a probability of $1-1/n^3$, where $p'=\mu_g(1-\alpha) + \mu_r\alpha$ and  $\gamma$ is a constant. \label{lem:atmost}
\end{Lemma}

Now, we analyze different scenarios of edge noise to understand the trade-off between required feedback and noise. 

\begin{Lemma}
For every pair of blocks, $B_c,B_d$  with more than $\gamma \log n$ records, the matching probability score estimate of $B_c$ with $1-\alpha$ fraction of matching edges is greater than the score of $B_d$ with $1-\beta$ (with $\alpha<\beta$) fraction of matching edges   with a  probability of $1-\frac{2}{n} $ if  $ \left((1-\alpha)\mu_g + \alpha \mu_r\right) >3 \left((1-\beta)\mu_g + \beta \mu_r\right)$
even in the absence of feedback.\label{lem:converge0feed}
\end{Lemma}
\begin{proof}
Using Lemma \ref{lem:atleast} and \ref{lem:atmost}, we can evaluate the condition that $score(B_c)>score(B_d)$ with a probability of $1-\frac{2}{n^3}$, in the absence of feedback. In order to guarantee this for all blocks, we perform a union bound over $\Theta(n^2)$ pairs of blocks, guaranteeing the success rate to $1-o(1)$.
\end{proof}

The previous lemma shows a scenario where the noise is not high and the prior based estimation of matching probability scores give a correct ordering of blocks. Now, we consider the more challenging noisy scenario and show that $\Theta(\log^2 n)$ feedback per block is enough for correct ordering.

\begin{Lemma}
For every pairs of blocks, $B_c,B_d$  with more than $\gamma \log n$ records, the matching probability score estimate of $B_c$ with $1-\alpha$ fraction of matching edges is greater than the score of $B_d$ with $1-\beta$ (where $\alpha<\beta$) fraction of matching edges  with a  probability of $1-\frac{2}{n} $ whenever  the ER phase provides overall feedback of $\Theta(n\log^2 n)$ randomly chosen edges.\label{lem:converge}
\end{Lemma}
\begin{proof}
Using Lemma \ref{lem:atmost}, $score(B_c)\ge {|F|}/{{\gamma \log n\choose 2}}(1-\alpha) + 0.5(\mu_g(1-\alpha)+\alpha\mu_r)(1-|F|/{\gamma \log n \choose 2}) $
and using Lemma \ref{lem:atleast}, $score(B_d) \le {|F|}/{{\gamma \log n \choose 2}}(1-\beta) + 1.5(\mu_g(1-\beta)+\beta\mu_r)(1-|F|/{\gamma \log n \choose 2})$
with a probability of $1-\frac{2}{n^3}$. Hence, $score (B_c) > score (B_d)$ holds if 
$F=c\log^2 n$, where $c$ is a large constant.
With a union bound over ${n\choose 2}$ pairs of blocks,  the score of any block $B_c$ (with higher fraction of matches) is higher than that of any block $B_d$ (with lower fraction of matches) with a probability of $1-\frac{2}{n}$. The total feedback to ensure $\Theta(\log^2 n)$ feedback on each block is $\Theta(n\log^2 n)$ as we consider $\Theta(n)$ blocks for scoring. 
\end{proof}


\myparagraph{Convergence for small blocks} The above analysis does not extend to blocks of size less than $\gamma \log n$. However,  all these blocks are ranked higher than the large blocks by TF-IDF. Hence, when \afro{} is initialized, the initial set of candidates generated will consider all these blocks before any of the larger blocks. In the worst case, there can be $\delta n$ such blocks, for some constant $\delta$ because our approach constructs a constant number of blocks per record (say $\delta$). Thus, the maximum number of candidates considered from small blocks is $\delta  n{\gamma\log n\choose 2}$ and all these candidates are considered in the first iteration of \afro{}. Following the discussion on small and large blocks, we  prove the main result of the convergence of \afro{}.

\begin{theorem}
 \afro{} pipeline achieves perfect recall with a feedback of $O (n\log^2 n)$ spread randomly across blocks.
\label{thm:main}
\end{theorem}
\begin{proof}
For blocks with more than $\gamma \log n$ records,  Lemmas \ref{lem:converge0feed} and \ref{lem:converge} show that a block with higher fraction of matching pairs is ranked higher than a block with fewer matching pairs, if provided with a feedback of $\Theta (n\log^2 n)$. 
Blocks with less than  $\gamma \log n$ records have not been considered above but in the worst case, these blocks generate $O(n\log^2 n)$ candidates as the maximum number of blocks considered is $\Theta(n)$. This ensures that a feedback of $\Theta (n\log^2 n)$ is sufficient to ensure the stated result. 
\end{proof}

\begin{table*}[ht!]
\centering
\scriptsize
\caption{Number of nodes $n$ (i.e., records), number of clusters $k$ (i.e., entities), size of the largest cluster $\left|C_1\right|$, the total number of matches in the data set $|E^+|$ and the reference to the paper where they appeared first. } \label{tab:dataset}
\begin{tabular}{ |l | | c |c| c | c  |  c | c | l |}
	\hline
dataset          & \multicolumn{2}{c|}{$n$}       & $k$       & $\left|C_1\right|$ & $\left|E^+\right|$ & ref.                                  & description \\ \hline\hline
\texttt{songs} & 1M & 1M       & ~0.99M        & 2                  & 146K                & \protect\cite{das2017falcon}       & Self-join of songs with very few matches. \\\hline
\texttt{citations}     & 1.8M & 2.5M        & \review{3.8M}        & 2                 & 558K                & \protect\cite{das2017falcon}       &   Bibliographic records from DBLP and CiteSeer. \\\hline
\texttt{products}   & 2554 & 22K       & \review{23.5K}     & $2$               & 1154              & \protect\cite{gokhale2014corleone}    & A collection of products from retail companies website.   \\\hline

\texttt{cora}    & \multicolumn{2}{c|}{1.9K}      & 191       & 236                & 62.9K              & \protect\cite{cora2004}              & Title, author, venue, and date of scientific papers. \\\hline
\texttt{cars}    & \multicolumn{2}{c|}{16.5K}       &  48       & 1799                 & 5.9M               & \protect\cite{carsVision}     & Descriptions of cars with make and model. 
\\\hline
\texttt{camera}    & \multicolumn{2}{c|}{29.7K}       &   26K      &         91         & 102K               & \protect\cite{di2kg}     & A collection of cameras from over 25 retail companies.  \\\hline
\end{tabular}
\end{table*}

\myparagraph{Discussion}
 Lemma \ref{lem:converge}  considers the convergence of block scores when the feedback is provided randomly over $\Theta(\log^2 n)$ edges within a block. If the feedback is biased towards $\Theta(\log^2 n)$ non-matching edges, the scores of noisier blocks will drop quicker and \afro{} will converge faster. Similarly, if the ER algorithm queries pairs with higher similarity (e.g. edge ordering \cite{wang2013leveraging}) or grows clusters by processing nodes (e.g. node ordering \cite{vesdapunt2014crowdsourcing}), providing larger feedback due to transitivity, this will only facilitate the growth (reduction) in score of  blocks with higher (lower) fraction of matching pairs leading to faster convergence.

%% file: experiments.tex

In this section we empirically demonstrate the ability of \afro{} to boost the efficiency and effectiveness of blocking and thus to improve the performance of ER. We also demonstrate the fast convergence of \afro{} thus confirming our theoretical analysis in Section~\ref{sec:perf}, and the robustness of \afro{} in different scenarios, including errors in ER results. This section is structured as follows.
\begin{compactitem}
\item \emph{Section~\ref{sec:exptwo}.} We compare the efficiency and effectiveness of \afro{} to prior work showing higher pair recall and faster running time in all the data sets. 
\item \emph{Section~\ref{sec:expone}.} We analyze \afro{} when used in conjunction with different ER methods showing higher \emph{F-score} (up to 60\%) irrespective of the method of choice. 
\item \emph{Section~\ref{sec:expthree}.} We study the dynamic performance of \afro{} and show its ability to converge monotonically to high effectiveness without compromising on efficiency in different scenarios including errors in ER results. 
\end{compactitem}

\subsection{Setup}
\label{sec:setup}

\noindent Before showing results we describe our experimental setup and the methods considered in our experiments.

\myparagraph{Experimental set-up} We implemented the algorithms in Java and machine learning tools in Python. The code runs on a server with 500GB RAM (all codes used $\leq$ 50GB RAM) and 64 cores. We consider six real-world data sets (see Table \ref{tab:dataset}) of various sizes and diverse cluster distributions. 
\dfnew{All the datasets are publicly available and come with their own manually curated ground truth.} We use publicly available pre-trained deep learning models\footnote{\url{https://cloud.google.com/vision}, \url{ https://www.ibm.com/watson/services/visual-recognition/}} to generate text descriptions of the image data (\texttt{cars}). For implementing the hierarchy we observed that we can trim at a depth of $10$ without any significant drop in the performance. 


\myparagraph{Blocking methods} 
We consider \dfnew{8} strategies for the blocking sub-tasks described in Section~\ref{sec:prel} and \dfnew{combine such strategies into 16 different pipelines. We study such pipelines with and without our \afro{} approach on top. 
\begin{compactitem}
\item[$\mathcal{BB}$)] We consider 4 methods for Block Building ($\mathcal{BB}$) and follow the suggestions of \cite{papadakis2018return} for their configuration. Standard blocking~\cite{papadakis2015schema} (\texttt{StBl}) generates a new block for each text token in the dataset. Q-grams blocking~\cite{gravano2001approximate} (\texttt{QGBL}) generates a new block for each 3-gram of characters. Sorted neighborhood~\cite{hernandez1995merge} (\texttt{SoNE}) sorts the tokens for each attribute and generates a new block for every sliding window of size 3 over these sort orders. Canopy clustering~\cite{mccallum2000efficient} (\texttt{CaCl}) generates a new block for each cluster of high similarity records (calculated as unweighted Jaccard similarity). We construct multiple instances of canopies (blocks), one for each attribute (i.e., based on the similarity of record pairs with respect to that attribute) and one based on all attributes together. 
\item[$\mathcal{BC}$)] We consider 2 traditional block scoring methods for Block Cleaning ($\mathcal{BC}$), dubbed \texttt{TF-IDF}~\cite{schutze2008introduction} and uniform scoring (\texttt{Unif}). For comparison purposes, we process blocks in non-increasing score order until the number of intra-block pairs equals to a parameter $M$ and then prune the remaining blocks. We set default $M$ to $10$ million.\footnote{\review{We note that setting a score threshold rather than a limit on the number of pairs would not take into account different scores distributions fairly.}} 
\item[$\mathcal{CC}$)] We consider 2 popular methods for Comparison Cleaning ($\mathcal{CC}$), dubbed meta-blocking~\cite{metablocking} (\texttt{MB}) and \bloss{}~\cite{dal2018bloss}, and follow the suggestions of \cite{metablocking} for their configuration. Weights of record pairs are set to their Jaccard similarity weighted with the block scores from the $\mathcal{BC}$ sub-task. We consider the top $100$ high-weight pairs for each record and prune the remaining record pairs. 
\end{compactitem}}

%
\noindent We recall that variants of our approach are denoted as \afro(,,) while traditional blocking pipelines without feedback are denoted as $\mathcal{B}$(,,) where the parameters correspond to techniques for $\mathcal{BB},\ \mathcal{BC}$ and $\mathcal{CC}$ sub-tasks, respectively. Default methods are \texttt{StBl} for $\mathcal{BB}$, \texttt{TF-IDF} for $\mathcal{BC}$ and \texttt{MB} for $\mathcal{CC}$. Default $\phi$ \dfnew{for \afro{}} is $0.01$.

\begin{figure*}[ht!]
\centering
\includegraphics[width= \textwidth]{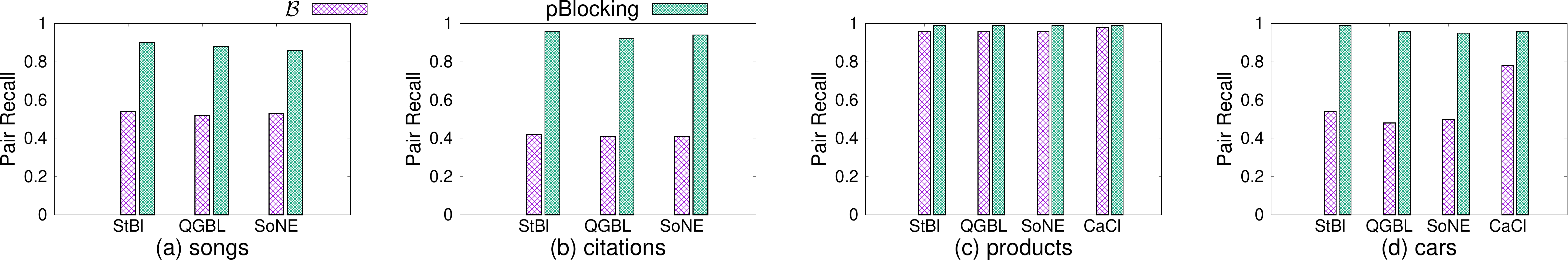}\\
\includegraphics[width= \textwidth]{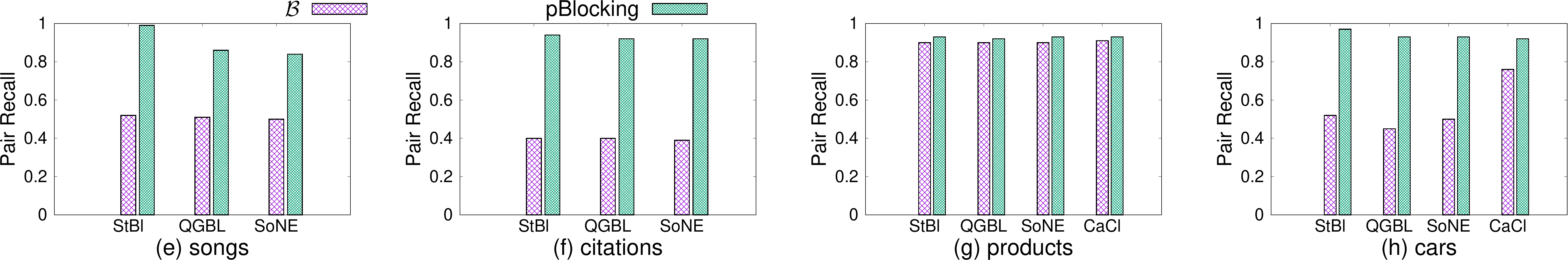}\\
\caption{Pair recall of $\mathcal{B}$(,\texttt{TF-IDF},) and \afro{}(,\texttt{TF-IDF},) with varying $\mathcal{BB}$ and $\mathcal{CC}$. (a-d) use \texttt{MB} and (e-h) use \texttt{BLOSS}. \texttt{CaCl} did not finish within 24 hrs on \texttt{songs} and \texttt{citations} data set. \label{fig:bbafro}}
\end{figure*}

\myparagraph{Pair matching and Clustering methods} We consider the following 3 strategies that leverage the notion of an \emph{oracle} to answer pairwise queries of the form ``does $u$ match with $v$?'' (a) $\texttt{Edge}$ ~\cite{wang2013leveraging} with default parameter setting. (b) \eager{}~\cite{galhotraSigmod}, the state-of-the-art technique to solve ER in the presence of erroneous oracle answers. (c) \texttt{Node} is the ER mechanism derived from~\cite{vesdapunt2014crowdsourcing} and was proposed as an improvement over \texttt{Edge}.
The \eager{} algorithm  handles noise for data sets with matching pairs much larger than $n$ and performs similar to \texttt{Edge} for data sets that have fewer matching pairs \cite{firmani2016online}, so we use it as default. 
We implement the abstract oracle tool with a classifier using scikit learn\footnote{\url{https://scikit-learn.org/stable/}} 
in Python. We consider two variants, Random forests (default) and a Neural Network. \review{The random forest classifier is trained with default settings of scikit learn. The neural network is implemented with a 3-layer convolutional neural network followed by two fully connected layers.  We  used word2vec word-embeddings  for each token in the records.  } In structured data sets, we extract similarity features for each attribute as in~\cite{das2017falcon}. For \texttt{cars} we use the text descriptions to calculate text-based features along with image-based features. Given the unstructured nature of text descriptions for some data sets we extracted POS tags using Spacy\footnote{\url{https://spacy.io/}}. All the considered classifiers are trained off-line with less than $1,000$ labelled pairs, \dfnew{containing a similar amount of matching and non-matching pairs. These labelled record pairs are the ones provided by the respective source for \texttt{citations}, \texttt{songs}, \texttt{products} and \texttt{camera} (the papers mentioned in Table~\ref{tab:dataset}, column ``ref.''). For \texttt{cars} and \texttt{cora} we perform active learning (following the guidelines of \cite{das2017falcon}) to identify a small set of labelled examples for training, which  are excluded from the evaluation of blocking quality.} 

\subsection{Benefits of Progressive Blocking}
\label{sec:exptwo}

In this experiment we evaluate the empirical benefit of \afro{} compared to previous blocking strategies. 


\myparagraph{Blocking effectiveness} Figure \ref{fig:bbafro} compares the Pair Recall (PR) of \afro{} and of a traditional blocking pipeline $\mathcal{B}$ for different choices of the block building and comparison cleaning techniques.  We use default block cleaning \texttt{TF-IDF} and default $M$ value. \afro{} achieves more than $0.90$ recall for all the data sets and with all the block building strategies, \review{ demonstrating its robustness to different cluster distributions and properties of the data.} Conversely, most of the considered block building strategies (\texttt{StBl}, \texttt{QGBL} and \texttt{SoNE}) have significantly lower recall even when used together with \bloss{} for selecting the pairs wisely. \texttt{QGBL} and \texttt{SoNE} help to improve recall in data sets with spelling errors but due to very few spelling mistakes in our data sets \texttt{StBl} has slightly higher recall. 
In terms of the data sets, the no-feedback blocking approach $\mathcal{B}$ has varied behavior. \texttt{products} and \texttt{camera} yield the best performance due to the presence of relatively cleaner blocks that help to easily identify matching pairs even without feedback. \texttt{songs} has higher noise and \texttt{cars} has a skewed distribution of clusters thereby making it harder for previous techniques to handle.  For this analysis, we do not consider \texttt{cora} (the smallest  data set) as it has less than $2$M pairs and hence, all techniques achieve perfect recall. We observed similar trends with \texttt{Unif} method for block cleaning in place of \texttt{TF-IDF} (discussed in Appendix). 

 
\begingroup
\setlength{\tabcolsep}{1pt} 
\begin{table}
\caption{\review{Running time comparison of $\mathcal{B}$(\texttt{StBl},\texttt{TF-IDF},\texttt{MB}) and \afro{}(\texttt{StBl},\texttt{TF-IDF},\texttt{MB}). \label{tab:timeafro}}}
\centering
\scriptsize
{
\begin{tabular}{ |l | | c | c |c|c|}
	\hline
	Dataset&\multicolumn{2}{c|}{0.95 Pair recall}  & \multicolumn{2}{c|}{Time budget: 1 hr} \\\hline
           &\afro{}     & $\mathcal{B}$ & \afro{}& $\mathcal{B}$    \\ \hline \hline
\texttt{songs}   &    29min  & 3hrs  & 0.96 &  0.78 \\
\texttt{citations}   &   55min    & Did not finish in 24 hrs& 0.97 & 0.64 \\
\texttt{cars}   &   4hr 10min    & 12hr  & 0.78& 0.54\\
\hline
\texttt{products}   &   6min 25sec  & 6min 13sec & 0.99&0.98    \\
\texttt{camera} &     12min     & 13min  & 0.97 &0.96 \\
\texttt{cora} &     5min 20 sec   & 5min 15 sec & 1 &  1  \\
\hline
\end{tabular}
}
\end{table}
\endgroup
\myparagraph{Blocking efficiency}
\review{In this experiment, we consider two different  settings  to compare (i) the time required to achieve more than $0.95$ pair recall (ii) the pair recall when the pipeline is allowed to run for a fixed amount of time ($1$ hour).  We run each technique for various values of $M$ and choose the best value that satisfies the required constraints. In the case of fixed budget of running time = $1$hour,  we run \afro{}'s feedback loop for the most iterations that allow the pipeline to process all records in the required time limit.
}

\review{Table~\ref{tab:timeafro} compares the total time required to achieve $0.95$ pair recall for each dataset\footnote{\review{This includes the time required by each approach to perform pair matching on the generated candidates.}}. \afro{} provides more than $3$ times reduction in running time for most large scale datasets in this setting.  In terms of total number of pairs enumerated, \afro{} considers around M=10 million to achieve $0.95$ recall for \texttt{citations} as opposed to more  than 200 million for $\mathcal{B}$.  We observed similar results for other block building (\texttt{SoNE}, \texttt{QGBL} and \texttt{CaCl}) and cleaning strategies. }

\review{The last two columns of Table~\ref{tab:timeafro} compare the pair recall of the generated candidates when the technique is allowed to run for $1$ hour.  \afro{} achieves better pair recall as compared to $\mathcal{B}$ across all datasets. The gain in recall is higher for larger datasets. The performance of \afro{} for \texttt{cars} is lower than that of \afro{} in Figure~2d because the feedback loop does not converge completely in 1hr. The pipeline runs for 8 rounds of feedback in this duration. This is consistent with the performance of \afro{} in Figure~4a, where the feedback is turned off after $10$ iterations.  }\review{The difference in performance of \afro{} and $\mathcal{B}$ is not high for small datasets of low noise like \texttt{products}, \texttt{cora} and \texttt{camera} as opposed to \texttt{songs}, \texttt{citations} and \texttt{cars}.
}

\subsection{Robustness of Progressive Blocking}
\label{sec:expone}
In this section, we evaluate the performance of \afro{}  with varying  strategies for pair matching and clustering in Algorithm~\ref{algo:afro} (referred to as $W$ in the pseudo-code). \dfnew{For this analysis, we use the default setting for $M$ as in Figure~\ref{fig:bbafro}.}

\myparagraph{Varying ER methods} \df{We recall that \afro{} can be used in conjunction with a variety of techniques for pair matching and clustering. }
Table \ref{tab:FL} compares the Pair Recall of the blocking graph, when using the different progressive ER methods mentioned in Section~\ref{sec:setup}. The final Pair Recall of \afro{} is more than $0.90$ in all data sets \review{and matching algorithms} except \texttt{citations} for node ER and more than $0.85$ in all cases.  This observation confirms our theoretical analysis in Section~\ref{sec:theory}, demonstrating that the feedback loop can improve the blocking, irrespective of the ER algorithm under consideration (which is a desirable property for a blocking algorithm). The above comparison of ER performance considers the algorithms with a default choice of  Random Forest classifier as the oracle.  
We observed that the feedback from the ER phase when using a Neural Network classifier contains slightly more errors but the blocking phase with \afro{} shows similar recall. We provide more discussion on ER errors in Section~\ref{sec:expthree}.


\myparagraph{Benefit on the final ER result} Table~\ref{tab:eager}  compares the  F-score of the final ER results when blocking is \review{performed} with and without \afro{}. In this experiment we use the state-of-the-art algorithm, \eager{} as the pair matching algorithm with default parameter values. Final F-score achieved with feedback is  more than 0.9 for all data sets except \texttt{products}. For \texttt{songs}, \texttt{citations} and \texttt{cars} the F-score of \afro{} is 1.5 times more than that of traditional blocking pipeline without feedback, thus demonstrating the effects of better effectiveness and efficiency of blocking. 

\begingroup
\setlength{\tabcolsep}{2pt} 
\begin{table}
\centering
\scriptsize
\caption{\dfnew{(a) Pair recall of \afro{} on varying ER strategies. (b) Comparison of the final F-score of the \texttt{Eager} method. The  blocking graph is computed with \afro{}(StBl, TF-IDF, MB) and $\mathcal{B}$(StBl, TF-IDF, MB) (both with default settings).}}
\vspace{-2mm}
{\subfloat[\label{tab:FL}]{\begin{tabular}{ |l | | c | c |c|c|}
	\hline
	    \multirow{2}{*}{Dataset} & \multirow{2}{*}{$\mathcal{B}$}\multirow{2}{*} & \multicolumn{3}{c|}{\afro{}} \\ \cline{3-5}
		    & & \texttt{Edge}      & \texttt{Node}   & \texttt{Eager}          \\
		\hline\hline
\texttt{songs}   &    0.53   & {0.9}  &  {0.9} & {0.9} \\
\texttt{citations}   &    0.42   & {0.90}  &0.87 & 0.95   \\
\texttt{cars}   &    0.54   & {0.98}   & 0.99 & 0.98  \\
\hline
\texttt{products}   &    0.95   & {0.98} & 0.98 & 0.98   \\
\texttt{camera} &     0.92     & 0.97 & 0.97 & 0.97  \\
\texttt{cora} &     1     & 1 & 1 & 1  \\
\hline
\end{tabular}}}
{
\subfloat[\label{tab:eager}]{
\begin{tabular}{ |p{1.2 cm} | | c | c|}
	\hline
\multirow{2}{*}{Dataset}  & \multirow{2}{*}{$\mathcal{B}$}     & \multirow{2}{*}{\afro{}}     \\
  &    &   \\\hline\hline
\texttt{songs} &   0.65  & \textbf{0.92}    \\
\texttt{citations}     &   {0.56}  & \textbf{0.92}   \\
\texttt{cars}     &   0.64  &   \textbf{0.94}    \\
\hline
\texttt{products}   &   0.71 &\textbf{0.72}     \\
\texttt{camera}     &   0.92  &  \textbf{0.95}   \\
\texttt{cora}     &   0.99  &  0.99     \\
\hline
\end{tabular}}}
\end{table}
\endgroup

\subsection{Progressive Behavior}
\label{sec:expthree}

This section studies the performance of \afro{} dynamically, in terms of (i) effect of feedback frequency $\phi$, (ii) effect of error on convergence, and 
(iii) convergence of the blocking result in the maximum number of rounds.

\begin{figure}
\centering
\subfloat[{Feedback Frequency }\label{fig:feedback}]{\includegraphics[width=0.24\textwidth]{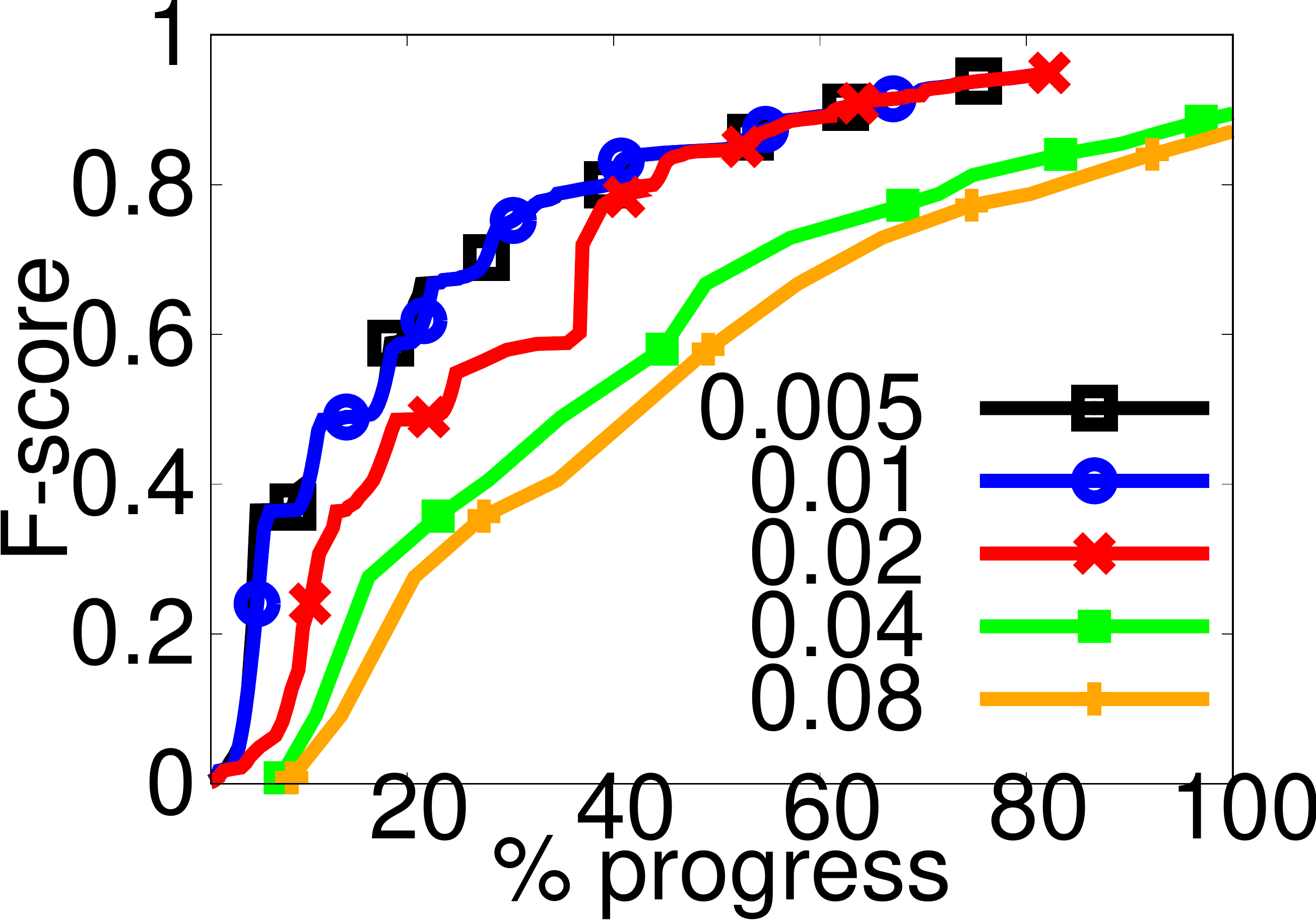}}
\subfloat[Oracle error \label{fig:error}]{\includegraphics[width=0.24\textwidth]{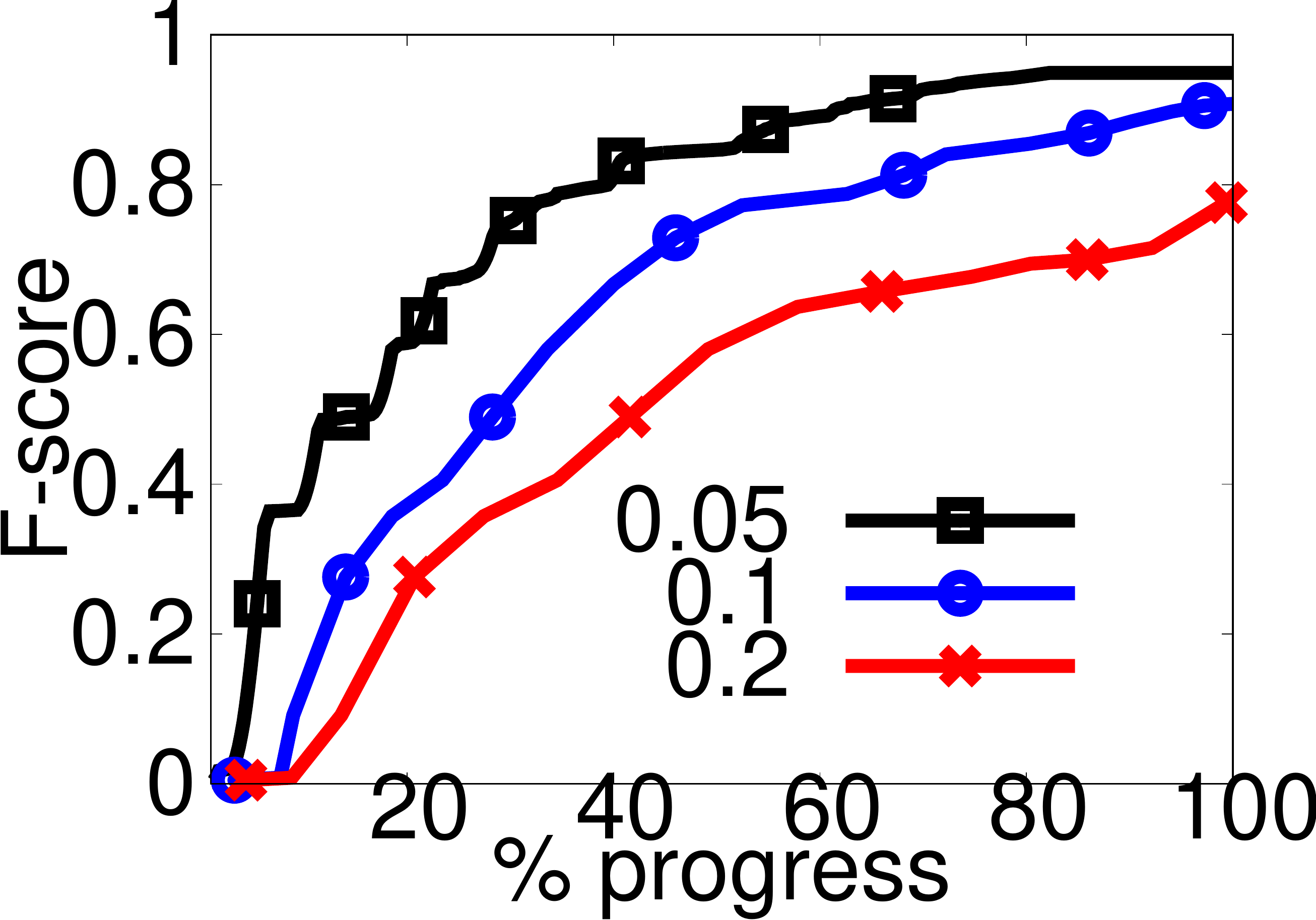}}

\caption{Progressive behavior of \afro{} with varying feedback frequency and errors in the feedback (\texttt{cars}).  \label{fig:freq}}
\end{figure}

\myparagraph{Feedback frequency} {
The $\phi$ parameter represents the fraction of newly processed record pairs after which feedback is sent from the partial ER results back to the blocking phase. Therefore, the parameter $\phi$ can control the maximum number of rounds of \afro{} and how often the blocking graph is updated. In order to describe the effect of varying $\phi$, Figure~\ref{fig:feedback} shows the F-score of the ER results as a function of the percentage of rounds completed, that we refer to as the \emph{blocking progress}.\footnote{Not to be confused with the ``ER progress'' in Algorithm~\ref{algo:afro}}. In the figure, different curves correspond to different feedback frequencies, including the default one (in blue). 
%
%
This plot shows that by updating the blocking graph more frequently (and thus increasing the number of rounds), the F-score \review{increases faster when $\phi$ is reduced from 0.08 to 0.01. The plot also shows that the F-score corresponding to smaller values of $\phi$ (up to 0.01) is consistently higher or equal as compared to the F-score corresponding to larger values of $\phi$.} 
Given that the running time of the pipeline increases with more frequent updates (\review{smaller values of $\phi$}), there appears to be limited value in  decreasing $\phi$ below $0.01$, thus justifying our choice for its default setting.
}

\myparagraph{Effect of ER errors} 
\review{As in} the previous experiment, Figure~\ref{fig:error} shows the effect of synthetic error in the ER results by varying the fraction of \dfnew{erroneous oracle answers. To this end, we corrupted the oracle answers randomly so as to get the desired amount of noise.} We note that even when 1 out of 5 answers are wrong, the final F-score is almost $0.8$, growing monotonically from the beginning to the end at the cost of a few extra pairs compared. \afro{} converges slower with higher error but the error does not accumulate and it performs much better than any other baseline. Additionally, we observed that even with 20\% error, the pair recall of \afro{} is as high as $0.98$ even though the F-score is close to 0.8 due to mistakes made by pair matching and clustering phase. This confirms that \afro{} is robust to errors in ER results and maintains high effectiveness to produce ER results with high F-score.

\begin{figure}
\centering
\includegraphics[width=0.48\textwidth]{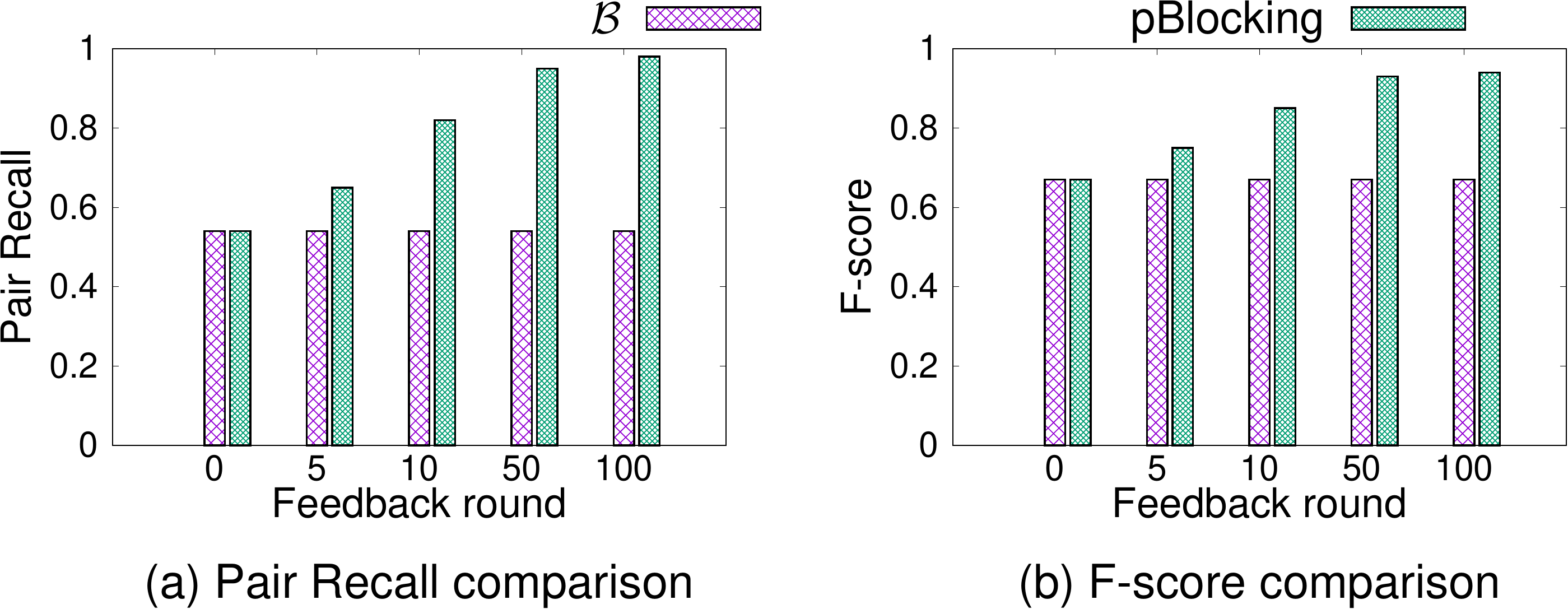}
\caption{Effect of feedback loop in  \texttt{cars} dataset.\label{fig:cars}}
\end{figure}
\myparagraph{Score Convergence} Figure 4a compares the Pair Recall (PR) of the blocking phase of \afro{}(\texttt{StBl},\texttt{TF-IDF},\texttt{MB}) after every round of feedback with the recall of $\mathcal{B}$(\texttt{StBl},\texttt{TF-IDF},\texttt{MB}). Both  $\mathcal{B}$ and \afro{} start with PR value close to 0.52 and \afro{}  consistently improves with more feedback achieving PR close to $0.9$ in less than 18 rounds. This shows the convergence of \afro{}'s score assignment strategy to achieve high PR values even with minimal feedback. Figure 4b compares the final F-score achieved by our method if the feedback loop is stopped after a few rounds. It shows that \afro{} achieves more than 0.8 F-score even when stopped after 10 rounds of feedback. This experiment validates that the convergence of block scoring leads to the convergence of the entire ER workflow.


\subsection{Key takeaways}
The empirical analysis in the previous sections has demonstrated \afro{}'s benefit on final F-score and its ability to boost effectiveness of blocking techniques across all data sets without compromising on efficiency. The key takeaways from our analysis are summarized below.  
\begin{compactitem}
     \item \afro{} improves Pair Recall irrespective of the technique used for block building, block cleaning or comparison cleaning (Figure \ref{fig:bbafro}), \df{thus demonstrating its flexibility}.
    
    \item Feedback based scoring helps in particular to boost blocking efficiency and effectiveness for noisy datasets with many matching pairs (i.e. containing large clusters) such as \texttt{cars}, by enabling accurate selection of cleanest blocks.
    
    \item The block intersection algorithm helps in particular with data sets with fewer matching pairs (i.e. with mainly small clusters) such as \texttt{citations} and \texttt{songs}, by providing a way to build small focused blocks with high fraction of matching pairs. 
    Block intersection can also help in data sets like \texttt{products} and \texttt{camera} but the benefit is not as high as that in \texttt{songs}, because many records in such data sets have unique identifiers (e.g. product model IDs) and thus initial blocks are reasonably clean.
\end{compactitem}

%% file: appendix.tex
\begin{figure*}
\centering
\includegraphics[width= \textwidth]{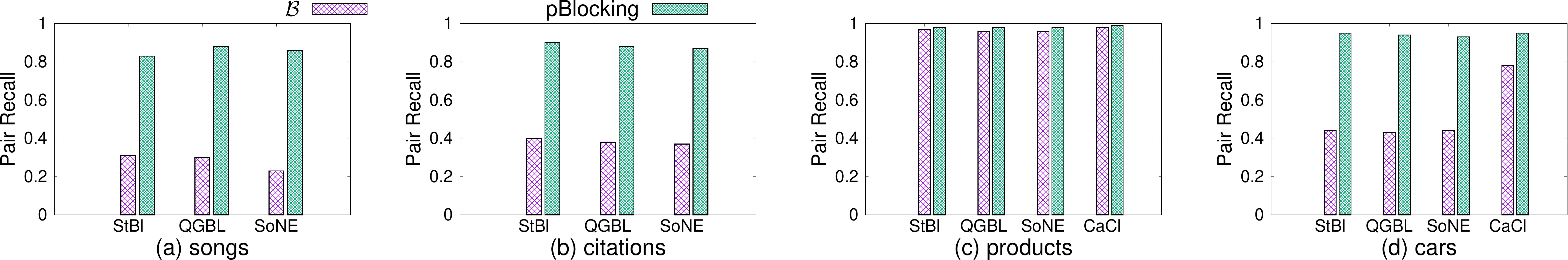}\\
\includegraphics[width= \textwidth]{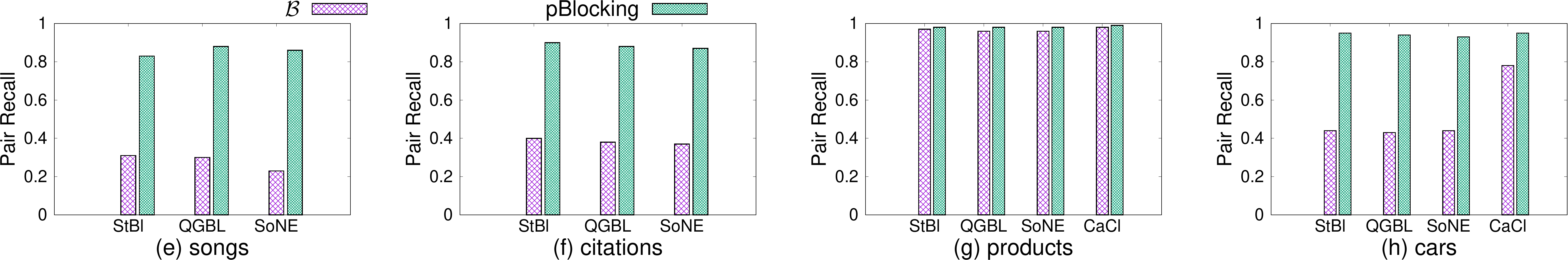}\\
\caption{Pair recall of $\mathcal{B}$(,\texttt{Unif},) and \afro{}(,\texttt{Unif},) with varying $\mathcal{BB}$ and $\mathcal{CC}$. (a-d) use \texttt{MB} and (e-h) use \texttt{BLOSS}. \texttt{CaCl} did not finish within 24 hrs on \texttt{songs} and \texttt{citations} data set. \label{fig:bbafroappendix}}
\end{figure*}

\section{Additional Experiments}

\myparagraph{Blocking Effectiveness} Figure~\ref{fig:bbafro} compares the Pair Recall of \afro{} and a traditional blocking pipeline $\mathcal{B}$, both with block-weights initialized with \texttt{TF-IDF} weighting mechanism. Figure~\ref{fig:bbafroappendix} performs the same comparison with the pipelines initialized using \texttt{Unif} weights.  Since, all blocks are assigned equal weight, we consider the block cleaning threshold of 100 along with default value of M. \afro{} performs substantially better than $\mathcal{B}$ for different settings of block building techniques across various datasets. With comparison to \texttt{TF-IDF} weighting scheme, \texttt{Unif} performs slightly worse but the difference is not substantial. The no-feedback pipeline $\mathcal{B}$ has varied performance across different data sets with the best performance on \texttt{products} and poorest performance on \texttt{citations} and \texttt{songs}.

